\documentclass[11pt, onecolumn]{IEEEtran}

\usepackage{amsmath,amssymb,mathrsfs,graphicx,cite}

\newtheorem{theorem}{Theorem}
\newtheorem{example}{Example}
\newtheorem{remark}{Remark}
\newtheorem{lemma}{Lemma}

\begin{document}
\title{Degrees of Freedom of Uplink--Downlink Multiantenna Cellular Networks}

\author{Sang-Woon Jeon\IEEEmembership{, Member, IEEE} and Changho Suh\IEEEmembership{, Member, IEEE}
\thanks{This research was funded by the MSIP (Ministry of Science, ICT $\&$ Future Planning), Korea in the ICT R$\&$D Program 2013.}
\thanks{The material in this paper was presented in part at the Information Theory and Applications Workshop (ITA), San Diego, CA, February 2014 and will be presented in part at the IEEE International Symposium on Information Theory (ISIT), Honolulu, HI, June/July 2014.
}
\thanks{S.-W. Jeon is with the Department of Information and Communication Engineering, Andong National University, South Korea (e-mail: swjeon@anu.ac.kr).}
\thanks{C. Suh is with the Department of Electrical Engineering, KAIST,  Daejeon, South Korea (e-mail: chsuh@kaist.ac.kr).}
}
\maketitle

\IEEEpeerreviewmaketitle


\begin{abstract}
An uplink--downlink two-cell cellular network is studied in which the first base station (BS) with $M_1$ antennas receives independent messages from its $N_1$ serving users, while the second BS with $M_2$ antennas transmits independent messages to its $N_2$ serving users. 
That is, the first and second cells operate as uplink and downlink, respectively.
Each user is assumed to have a single antenna.
Under this uplink--downlink setting, the sum degrees of freedom (DoF) is completely characterized as the minimum of $(N_1N_2+\min(M_1,N_1)(N_1-N_2)^++\min(M_2,N_2)(N_2-N_1)^+)/\max(N_1,N_2)$, $M_1+N_2,M_2+N_1$, $\max(M_1,M_2)$, and $\max(N_1,N_2)$, where $a^+$ denotes $\max(0,a)$.
The result demonstrates that, for a broad class of network configurations, operating one of the two cells as uplink and the other cell as downlink can strictly improve the sum DoF compared to the conventional uplink or downlink operation, in which both cells operate as either uplink or downlink. 
The DoF gain from such uplink--downlink operation is further shown to be achievable for heterogeneous cellular networks having hotspots and with delayed channel state information.
\end{abstract}


\section{Introduction} \label{sec:intro}
Characterizing the capacity of cellular networks is one of the fundamental problems in network information theory. 
Unfortunately, even for the simplest setting consisting of two base stations (BSs) having one serving user each, which is referred to as the two-user interference channel (IC), capacity is not completely characterized for general channel parameters \cite{Han:87,Etkin:08}.
Exact capacity results being notoriously difficult to obtain, many researchers have recently studied approximate capacity characterizations in the shape of so-called ``degrees of freedom (DoF)'', which captures the behavior of capacity as the signal-to-noise ratio (SNR)  becomes large.

The DoF metric has received a great deal of attention and thoroughly analyzed as multiantenna techniques emerged  \cite{FoschiniGans:98,Telatar:99}, especially in cellular networks \cite{Caire:03,Sriram:03,Viswanath:03,Yu:04,Weingarten:06} because of their potential to increase the DoF of cellular networks. 
Roughly speaking, equipping multiple antennas at the BS and/or users can drastically increase the sum DoF of single-cell cellular networks proportionally with the number of equipped antennas. 

Under multicell environment, Cadambe and Jafar recently made a remarkable progress showing that the optimal sum DoF for the $K$-user IC is given by $K/2$ \cite{Viveck1:08}, which corresponds to the $K$-cell cellular network having one serving user in each cell.
A new interference mitigation paradigm called interference alignment (IA) has been proposed to achieve the sum DoF $K/2$ \cite{Viveck1:08}.
Multicell cellular networks having multiple serving users in each cell has been studied in \cite{Suh:08,Suh:11} under both uplink and downlink operation, each of which is called interfering multiple access channel (IMAC) \cite{Suh:08} and  interfering broadcast channel (IBC) \cite{Suh:08,Suh:11}.
It was shown in \cite{Suh:08,Suh:11} that multiple users in each cell is beneficial for increasing the sum DoF of IMAC and IBC by  utilizing multiple users in each cell for IA.

As a natural extension, integrating multiantenna techniques and IA techniques has been recently studied to boost the DoF of multicell multiantenna cellular networks. 
The DoF of the $K$-user IC having $M$ antennas at each transmitter and $N$ antennas at each receiver has been analyzed in \cite{Tiangao:10}.
More recently, the IMAC and IBC models have been extended to multiantenna BS and/or multiantenna users, see \cite{Kim:11,Hwang:12,Shin:11,Shin:13,Liu2:13,Sridharan:13} and the references therein. 

\begin{figure}[t!]
\begin{center}
\includegraphics[scale=0.7]{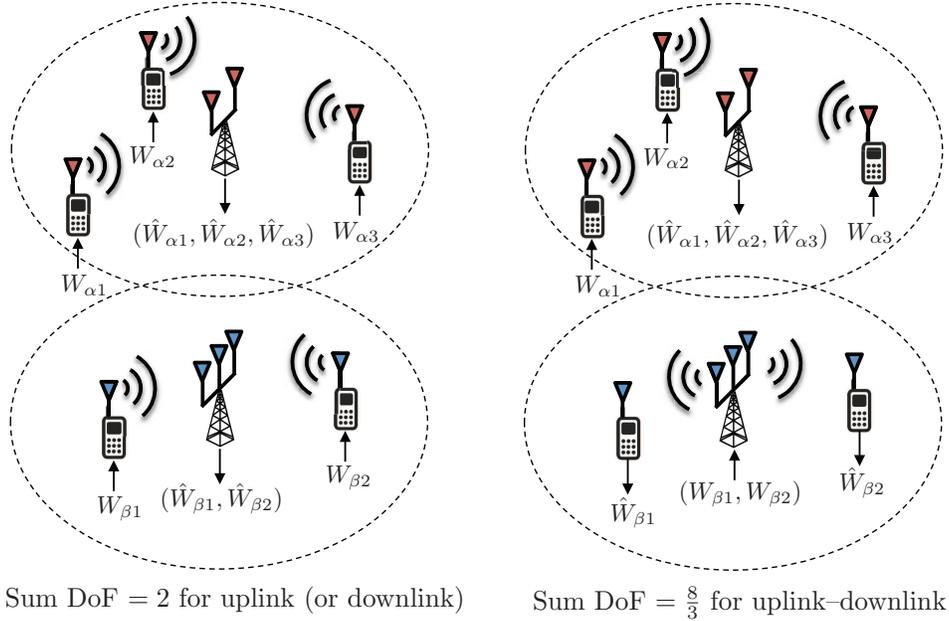}
\caption{Sum DoF of the example network.}
\label{figs:motivating_ex}
\end{center}
\end{figure}

\subsection{Motivating Example}
In this paper, we study a multiantenna two-cell cellular network in which the first and second cells operate as uplink and downlink respectively. For better understanding on the motivation of the paper, we introduce a simple two-cell cellular network in Fig. \ref{figs:motivating_ex}. The first cell consists of a BS having two antennas and three users but the second cell consists of a BS having three antennas and two users.
Let us consider how to operate or coordinate this example network in order to maximize its sum DoF.
As we will explain later, if both cells operate as the conventional uplink or downlink, then the sum DoF is limited by two from the DoF result of the two-user multiple input multiple output (MIMO) IC in \cite{Jafar:07}.
Hence, activating one of the two cells can trivially achieve the optimal sum DoF for these cases.
Notice that the another option is to operate the first cell as uplink and the second cell as downlink or vice versa.
For this case, the two-user MIMO IC upper bound in \cite{Jafar:07} is given by three, suggesting that it might be possible to achieve more than two sum DoF.
But it is at least impossible to achieve more than two DoF by simply activating one of two cells.
We will show that for this case the optimal sum DoF is given by $8/3$, strictly greater than that achievable by the conventional uplink or downlink operation.

The previous work on the DoF of multiantenna cellular networks, however, inherently assumes either uplink or downlink  so that it cannot capture the possibility of such DoF improvement from the uplink--downlink operation.
Therefore, the primary aim of this paper is to figure out whether operating as either the conventional uplink or downlink is optimal or not in terms of the DoF for multicell multiantenna cellular networks. 
We focus on two-cell networks in which the first cell, consisting of a BS with $M_1$ antennas and $N_1$ users, operates as uplink and the second cell, constisting of a BS with $M_2$ antennas and $N_2$ users, operates as downlink.
We completely characterize the sum DoF and the result demonstrates that, depending on the network configuration, uplink--downlink operation is beneficial for increasing the sum DoF compared to the conventional  uplink or downlink operation.

\subsection{Previous Work}

In seminal work \cite{Viveck1:08}, Cadambe and Jafar showed that the optimal sum DoF of the $K$-user IC with time-varying channel coefficients is given by $\frac{K}{2}$, achievable by signal space IA.
The concept of this signal space alignment has been successfully adapted to various network environments, e.g., see \cite{Maddah-Ali:08,Viveck2:09,Tiangao:10,Suh:11, Suh:08,Viveck1:09,Annapureddy:11,Ke:12} and the references therein.
It was shown in \cite{Motahari:09, Motahari2:09} that IA can also be attained on fixed (not time-varying) channel coefficients.  
A different strategy of IA was developed in \cite{Nazer11:09,Jeon:13} called ergodic IA, which makes interference aligned in the finite SNR regime and, as a result, provides significant rate improvement compared with the conventional time-sharing strategy in the finite SNR regime \cite{Nazer11:09,Jeon2:14}.
The DoF of $K$-user {MIMO} IC has been considered in \cite{Tiangao:10,Yetis:10,Wang:12}.

For multisource multihop networks, interference can not only be aligned, but it can be cancelled through multiple paths, which is referred to as interference neutralization \cite{Rankov:07}.
The work \cite{Tiangao:12} has exploited IA to neutralize interference at final destinations, which is referred to as aligned interference neutralization, and showed that the optimal sum DoF two is achievable for $2$-user $2$-hop networks with $2$ relays.
Similar concept of ergodic IA has been proposed for interference neutralization in \cite{Jeon2:11} showing that ergodic interference neutralization achieves the optimal sum DoF of $K$-user $K$-hop isotropic fading networks with $K$ relays in each layer.
Recently, it has been shown in \cite{Shomorony:13} that the optimal sum DoF of the $K$-user $2$-hop network with $K$ relays is given by $K$.

The DoF of cellular networks has been first studied by Suh and Tse for both uplink and downlink environments, called IMAC and IBC respectively \cite{Suh:08,Suh:11}. 
It was shown that, for two-cell networks having $K$ users in each cell, the sum DoF $\frac{2K}{K+1}$ is achievable for both uplink and downlink. Hence, multiple users at each cell are beneficial for improving the DoF of cellular networks.
The IMAC and IBC models have been extended to have multiple antennas at each BS and/or user \cite{Kim:11,Zhuang:11,Pantisano:11,Guillaud:11,Shin:11,Hwang:12,Liu2:13,Shin:13,Liu:13,Sridharan:13,Park:12,Ntranos:14,Shin3:12,Ayoughi:13}.
For multiantenna IMAC and IBC, it was shown that there exists in general a trade-off between two approaches: zero-forcing by using multiple antennas and asymptotic IA by treating each antenna as a separate user \cite{Wang:12,Park:12, Liu2:13,Sridharan:13}.

Recently, reverse time division duplex (TDD), i.e.,  operating a subset of cells as uplink and the rest of the cells as downlink, has been actively studied in heterogeneous cellular networks, consisting of macro BSs with larger number of antennas and micro BSs with smaller number of antennas \cite{Ghosh:12,Kountouris:13,Hoydis:13,Hosseini:13,Andrews:13,Adhikary:14}. Under various practical scenarios, potential benefits of reverse TDD have been analyzed in the context of coverage \cite{Kountouris:13}, area spectral efficiency \cite{Kountouris:13,Hoydis:13}, throughput \cite{Hosseini:13,Adhikary:14}, and so on.

\subsection{Paper Organization}
The rest of this paper is organized as follows.
In Section \ref{sec:prob_formulation}, we introduce the uplink--downlink multiantenna two-cell cellular network model and define its sum DoF. In Section \ref{sec:main_result}, we first state the main result of this paper, the sum DoF of the uplink--downlink multiantenna two-cell cellular network. The proof of the main result is presented in Section \ref{sec:achievability}.
We then discuss some related problems regarding the main result in Section \ref{sec:discussion} and finally conclude in Section \ref{sec:conclusion}.

\section{Problem Formulation} \label{sec:prob_formulation}
We will use boldface lowercase letters to denote vectors and boldface uppercase letters to denote matrices.
Throughout the paper, $[1:n]$ denotes $\{1,2,\cdots,n\}$, $\mathbf{0}_n$ denotes the $n\times 1$ all-zero vector, and $\mathbf{I}_n$ denotes the $n\times n$ identity matrix. 
For a real value $a$, $a^+$ denotes $\max(0,a)$.
For a set of vectors $\{\mathbf{a}_i\}$, $\operatorname{span}(\{\mathbf{a}_i\})$ denotes the vector space spanned by the vectors in $\{\mathbf{a}_i\}$.
For a vector $\mathbf{b}$, $\mathbf{b}\perp\operatorname{span}(\{\mathbf{a}_i\})$ means that $\mathbf{b}$ is orthogonal with all vectors in $\operatorname{span}(\{\mathbf{a}_i\})$.
For a matrix $\mathbf{A}$, $\mathbf{A}^{\dagger}$ denotes the transpose of $\mathbf{A}$. 
For a set of matrices $\{\mathbf{A}_i\}$, $\operatorname{diag}(\mathbf{A}_1,\cdots, \mathbf{A}_n)$ denotes the block diagonal matrix consisting of $\{\mathbf{A}_i\}$.

\begin{figure}[t!]
\begin{center}
\includegraphics[scale=0.7]{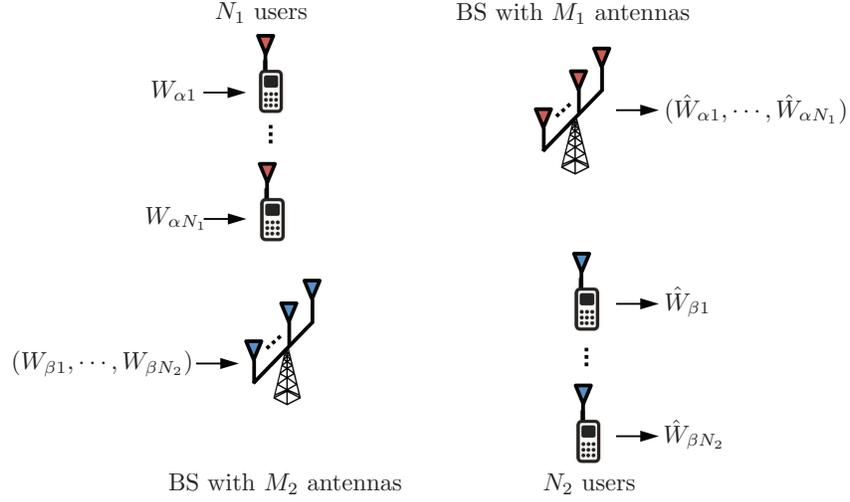}
\caption{Uplink--downlink multiantenna two-cell cellular networks, where the transmitters are located in the left hand side by convention.}
\label{figs:general_i_bc_mac}
\end{center}
\end{figure}

\subsection{Uplink--Downlink Multiantenna Two-Cell Cellular Networks}
Consider a multiantenna two-cell cellular network depicted in Fig. \ref{figs:general_i_bc_mac} in which the first cell (cell $\alpha$) operates as uplink and the second cell (cell $\beta$) operates as downlink.
Specifically, the BS in cell $\alpha$ (BS $\alpha$) equipped with $M_1$ antennas wishes to receive an independent message $W_{\alpha i}$ from the $i$th user in the same cell (user $(\alpha,i)$) for all $i\in[1:N_1]$. 
On the other hand, the BS in cell $\beta$ (BS $\beta$) equipped with $M_2$ antennas wishes to send an independent message $W_{\beta j}$ to the $j$th user in the same cell  (user $(\beta,j)$) for all $j\in[1:N_2]$.
Each user is assumed to have a single antenna.

The $M_1\times 1$ received signal vector of BS $\alpha$ at time $t$ is given by
\begin{equation} \label{eq:in_out_1}
\mathbf{y}_{\alpha}[t]=\sum_{i=1}^{N_1}\mathbf{h}_{\alpha i}[t]x_{\alpha i}[t]+\mathbf{G}_{\alpha}[t]\mathbf{x}_{\beta}[t]+\mathbf{z}_{\alpha}[t]
\end{equation}
and the received signal of user $(\beta,j)$ at time $t$ is given by
\begin{equation} \label{eq:in_out_2}
y_{\beta j}[t]=\mathbf{h}_{\beta j}[t]\mathbf{x}_{\beta}[t]+\sum_{i=1}^{N_1}g_{\beta ji}[t]x_{\alpha i}[t]+z_{\beta j}[t],
\end{equation}
where $j\in[1:N_2]$.
Here $\mathbf{h}_{\alpha i}[t]\in\mathbb{R}^{M_1\times 1}$ is the channel vector from user $(\alpha,i)$ to BS $\alpha$, $\mathbf{G}_{\alpha}[t]\in\mathbb{R}^{M_1\times M_2}$ is the channel matrix from BS $\beta$ to BS $\alpha$, $\mathbf{h}_{\beta j}[t]\in\mathbb{R}^{1\times M_2}$ is the channel vector from BS $\beta$ to user $(\beta,j)$, and $g_{\beta ji}[t]\in\mathbb{R}$ is the scalar channel from user $(\alpha,i)$ to user $(\beta,j)$.
Also, $x_{\alpha i}[t]\in \mathbb{R}$ is the transmit signal of user $(\alpha,i)$ and  $\mathbf{x}_{\beta}[t]\in\mathbb{R}^{M_2\times 1}$ is the transmit signal vector of cell $\beta$.
The additive noise vector at cell $\alpha$, denoted by  $\mathbf{z}_{\alpha}[t]\in \mathbb{R}^{M_1\times 1}$, is assumed to follow 
$\mathcal{N}(\mathbf{0}_{M_1},\mathbf{I}_{M_1})$ .
Similarly, the additive noise at user $(\beta,j)$, denoted by $z_{\beta j}[t]$, is assumed to follow $\mathcal{N}(0,1)$.
Each user in cell $\alpha $ and BS $\beta$ should satisfy the average power constraint $P$, i.e., $E\big(x^2_{\alpha i}[t]\big)\le P$ for all $i\in[1:N_1]$ and $E\left(\|\mathbf{x}_{\beta}[t]\|^2\right)\le P$, where $\|\cdot\|$ denotes the norm of a vector.

We assume that all channel coefficients are independent and identically distributed (i.i.d.) drawn from a continuous distribution and vary independently over each time slot. Global channel state information (CSI) is assumed to be available at each user and BS. 

\subsection{Degrees of Freedom}
Let $W_{\alpha i}$ and $W_{\beta j}$ be chosen uniformly at random from $[1:2^{nR_{\alpha i}}]$ and $[1:2^{nR_{\beta j}}]$ respectively, where $i\in[1:N_1]$ and $j\in[1:N_2]$.
A rate tuple $(R_{\alpha 1},\cdots,R_{\alpha N_1},R_{\beta 1},\cdots,R_{\beta N_2})$ is said to be achievable if there exists a sequence of $(2^{nR_{\alpha 1}},\cdots,2^{nR_{\alpha N_1}},2^{nR_{\beta 1}},\cdots,2^{nR_{\beta N_2}};n)$ codes such that $\Pr(\hat{W}_{\alpha i}\neq W_{\alpha i})\to 0$ and $\Pr(\hat{W}_{\beta j}\neq W_{\beta j})\to 0$ as $n$ increases for all $i\in[1:N_1]$ and $j\in[1:N_2]$.
Then the achievable sum DoF is given by
\begin{equation}
\lim_{P\to\infty}\frac{\sum_{i=1}^{N_1} R_{\alpha i}+\sum_{j=1}^{N_2} R_{\beta j}}{\frac{1}{2}\log P}.
\end{equation}
For notational convenience, denote the  maximum achievable sum DoF by $d_{\Sigma}$.
In the rest of the paper, we will characterize $d_{\Sigma}$, which is given by a function of $M_1$, $M_2$, $N_1$, and $N_2$.

\section{Main Result} \label{sec:main_result}
In this section, we state our main result.
We completely characterize $d_{\Sigma}$ in the following theorem.

\begin{theorem} \label{thm:achievable_DoF}
For the uplink--downlink multiantenna two-cell cellular network, 
\begin{align} \label{eqn:achievable_DoF}
d_{\Sigma}=\min\Bigg\{&\frac{N_1N_2+\min(M_1,N_1)(N_1-N_2)^+ +\min(M_2,N_2)(N_2-N_1)^+}{\max (N_1,N_2)},\nonumber\\
&M_1+N_2,M_2+N_1,\max(M_1,M_2),\max(N_1,N_2)\Bigg\}.
\end{align}
\end{theorem}
\begin{proof}
We refer to Section \ref{sec:achievability} for the proof.
\end{proof}

For better understanding of the contribution of the main result, we present simple existing upper and lower bounds on $d_{\Sigma}$.
Obviously, $d_{\Sigma}$ is upper bounded by the sum DoF of the two-user MIMO IC having $N_1$ transmit antennas and $M_2$ received antennas for the first transmission pair and $M_2$ transmit antennas and $N_2$ received antennas for the second transmission pair. Hence, from the result in \cite{Jafar:07},
\begin{align} \label{eq:two_user_IC_bound}
d_{\Sigma}\leq\min\{M_1+N_2,M_2+N_1,\max(M_1,M_2),\max(N_1,N_2)\}.
\end{align}
Note that the first DoF constraint in \eqref{eqn:achievable_DoF} do not appear in \eqref{eq:two_user_IC_bound}, which can be interpreted as the DoF degradation due to distributed processing at each user. On the other hand, if only one of the two cells is activated, we have
\begin{align} \label{eq:single_cell_bound}
d_{\Sigma}\geq \max(\min(M_1,N_1),\min(M_1,N_2)).
\end{align}

In the following, we first consider symmetric cell configurations in which either the number of antennas at each BS or the number of users in each cell is the same. For this case, $d_{\Sigma}$ is trivially characterized from \eqref{eq:two_user_IC_bound} and \eqref{eq:single_cell_bound} without using Theorem \ref{thm:achievable_DoF}.

\begin{example}[Symmetric Cell Configurations] \label{ex:symmetric_case}
First consider the case where the number of antennas at each BS is the same, i.e., $M_1=M_2:=M$. 
Then the existing upper and lower bounds in \eqref{eq:two_user_IC_bound} and \eqref{eq:single_cell_bound} coincide showing that $d_{\Sigma}=\min(M,\max(N_1,N_2))$ for this case. 
The same is true for the case where the number of users in each cell is the same, i.e., $N_1=N_2:=N$. Then $d_{\Sigma}=\min(\max(M_1,M_2),N)$. \hfill$\lozenge$
\end{example}

For a general (asymmetric) cell configuration, however, the upper and lower bounds in \eqref{eq:two_user_IC_bound} and \eqref{eq:single_cell_bound} is not tight as demonstrated in the following example.

\begin{figure}[t!]
\begin{center}
\includegraphics[scale=0.5]{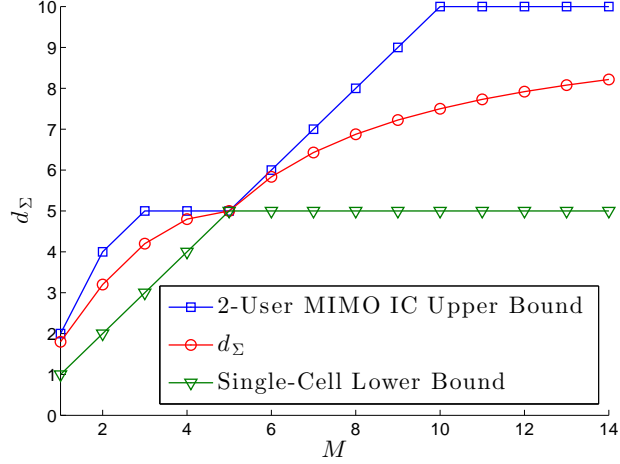}
\end{center}
\vspace{-0.15in}
\caption{$d_{\Sigma}$ in Theorem \ref{thm:achievable_DoF} with respect to $M$ when $N=5$, where $M_1=N_2=M$ and $M_2=N_1=N$.}
\label{figs:DoF_N6}
\vspace{-0.1in}
\end{figure}

\begin{example}[Asymmetric Cell Configurations] \label{ex:asymmetric_case}
Consider the asymmetric cell configuration in which $M_1=N_2:=M$ and $M_2=N_1:=N$.
Then Theorem \ref{thm:achievable_DoF} shows that 
\begin{align} \label{eq:DoF_symmetric_case}
d_{\Sigma}=\begin{cases}
\frac{M(2N-M)}{N} &\mbox{ if } M\le N,\\
\frac{N(2M-N)}{M} &\mbox{ if } M> N.
\end{cases}
\end{align}
Figure \ref{figs:DoF_N6} plots \eqref{eq:DoF_symmetric_case} with respect to $M$ when $N=5$.
For comparison, we also plot the two-user MIMO IC upper bound \eqref{eq:two_user_IC_bound} and the single-cell lower bound \eqref{eq:single_cell_bound}, each of which is given by $\min\{2M,2N,\max(M,N)\}$ and $\min(M,N)$ respectively.
Note that \eqref{eq:DoF_symmetric_case} is not trivially achievable and, moreover, the two-user MIMO IC upper bound is not tight for all $M$ and $N$ satisfying $M\neq N$. \hfill$\lozenge$
\end{example}

The above two examples have led to a fundamental question: \emph{Which class of cell configurations can  uplink--downlink operation improve the sum DoF of cellular networks compared to the conventional uplink or downlink operation (including the single-cell operation)?} 
That is, the question is about the \emph{cell coordination problem} when a network is able to choose the operation mode of each cell to maximize its sum DoF.
For a broad class of heterogeneous cell configurations, uplink--downlink operation strictly improves the sum DoF compared to the case where the entire cells operate either uplink or downlink.
We briefly address this question in the following remark based on the cell configuration assumed in Example \ref{ex:symmetric_case}.
The DoF gain from uplink--downlink operation will be discussed in more details over a general four-parameter space $(M_1,M_2,N_1,N_2)$ in Section \ref{subsec:dof_gain}.
We further address the above question for cellular networks having hotspots in Section \ref{subsec:dof_hetnet}, which is a certain type of heterogeneous cellular networks.

\begin{figure}[t!]
\begin{center}
\includegraphics[scale=0.7]{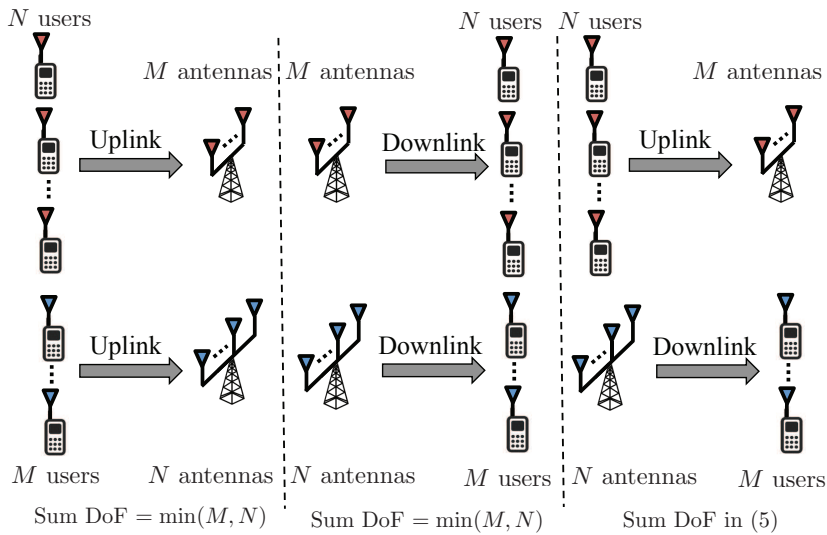}
\end{center}
\vspace{-0.15in}
\caption{Sum DoF achievable by uplink, downlink, and uplink--downlink operation for $M_1=N_2=M$ and $M_2=N_1=N$.}
\label{figs:DoF_gain}
\vspace{-0.1in}
\end{figure}

\begin{remark}[Dof Gain From Uplink--Downlink Operation] \label{re:uplink_downlink}
Theorem \ref{thm:achievable_DoF} demonstrates that, depending on the network configuration, operating one cell as unlink and the other cell as downlink improves the sum DoF compared to the conventional operation in which the entire cells operate as either uplink or downlink. 
For instance, consider the cell coordination problem for the two-cell heterogeneous cellular network in which its configuration is given as in Fig. \ref{figs:general_i_bc_mac}. That is, the operation mode of each cell can be coordinated to maximize the sum DoF.
As shown in Fig. \ref{figs:DoF_gain}, if we operate both cells either uplink or downlink, then the sum DoF is upper bounded by the single-cell lower bound, i.e., $\min(M,N)$.
On the other hand, uplink--downlink operation achieves \eqref{eq:DoF_symmetric_case}, which is strictly larger than $\min(M,N)$ for all  $M$ and $N$ satisfying $M\neq N$.
Furthermore, the DoF gain from uplink--downlink operation becomes significant as the difference between $M$ and $N$ increases. Specifically, $d_{\Sigma}\to 2M$ as $N\to\infty$ in \eqref{eq:DoF_symmetric_case}. Whereas the sum DoF achievable by the conventional uplink or downlink operation is limited by $M$ even as $N\to\infty$. \hfill$\lozenge$
\end{remark}

The following remark states an interesting observation captured by Theorem \ref{thm:achievable_DoF}.
It is about the impact of user cooperation on the two-cell IMAC or IBC, which corresponds to the model assuming the conventional uplink or downlink and, thus, is not related to uplink--downlink operation.

\begin{figure}[t!]
\begin{center}
\includegraphics[scale=0.7]{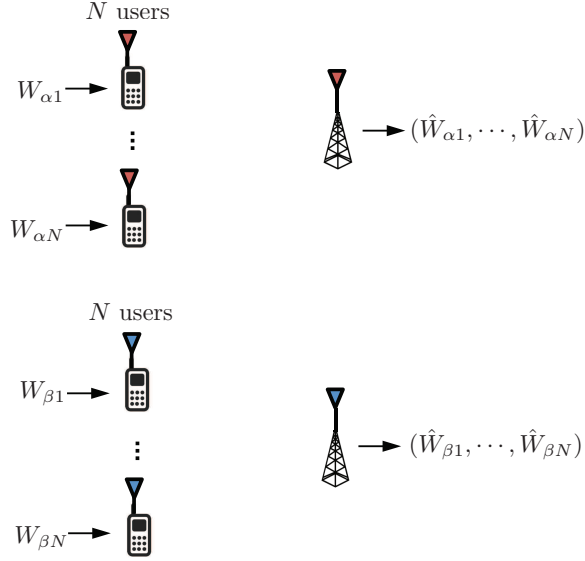}
\end{center}
\vspace{-0.15in}
\caption{Two-cell IMAC in which each BS and user is equipped with a single antenna.}
\label{figs:user_cooperation}
\vspace{-0.1in}
\end{figure}

\begin{remark}[User Cooperation]
Consider  the two-cell IMAC in Fig. \ref{figs:user_cooperation} in which $N$ users in each cell wish  to transmit independent messages to their BS.
Suh and Tse showed that the sum DoF $\frac{2N}{N+1}$ is achievable in this case, which converges to the interference-free sum DoF of $2$ as $N$ increases.
Obviously, if the users within each cell can cooperate with each other, then the interference-free sum DoF is achievable if $N\geq 2$. Hence the number of users in each cell does not have to go to infinity.
Now suppose that the users in the second cell can cooperate.  
From Theorem \ref{thm:achievable_DoF}, $d_{\Sigma}=\frac{2N_1-1}{N_1}$ when $M_2=2$ and $M_1=N_2=1$, which shows $d_{\Sigma}\to 2$ as $N_1\to \infty$.
Hence this result shows that, even though user cooperation is allowed only for the second cell, cooperation between two users is enough to achieve $d_{\Sigma}\to 2$ if the number of users in the first cell tends to infinity. 
In this sense, one-side user cooperation is still powerful for boosting DoF.
The same argument holds for the two-cell IBC. \hfill$\lozenge$
\end{remark}

\section{Proof of Theorem \ref{thm:achievable_DoF}} \label{sec:achievability}
In this section, we prove Therorem \ref{thm:achievable_DoF}. 
We first provide the converse proof  in Section \ref{subsec:converse} and then provide the achievability proof in Sections \ref{subsec:main_idea} to \ref{subsec:scheme2}.
For better understanding of the achievability idea, we first establish it  based on a simple example network in Section \ref{subsec:main_idea}. 
We then introduce two proposed schemes for a general network and analyze their achievable sum DoF in Sections \ref{subsec:achievable_dof} to \ref{subsec:scheme2}. 

\subsection{Converse} \label{subsec:converse}
In this subsection, we prove the converse of Theorem \ref{thm:achievable_DoF}.
If full cooperation is allowed within the $N_1$ users in cell $\alpha$ and within the $N_2$ users in cell $\beta$, then the network becomes the two-user MIMO IC. Hence, 
$d_{\Sigma}\leq \min\{M_1+N_2,M_2+N_1,\max(M_1,M_2),\max(N_1,N_2)\}$ from the result in \cite{Jafar:07}.
Then the remaining part is to prove the first $d_{\Sigma}$ constraint in \eqref{eqn:achievable_DoF}.

Denote $d_{\alpha i}$, $i\in[1:N_1]$ by an achievable DoF of user $(\alpha,i)$ and $d_{\beta j}$, $j\in[1:N_2]$, by an achievable DoF of user $(\beta,j)$. 
Let us then remove all the users in cell $\alpha$ except user $(\alpha, i)$ and all the users in cell $\beta$ except user $(\beta, j)$.
Obviously, removing other users cannot degrade $d_{\alpha i}+d_{\beta j}$.
Therefore, again from \eqref{eqn:achievable_DoF}, 
\begin{align} \label{eq:converse_lemma}
d_{\alpha i}+d_{\beta j}\leq 1.
\end{align}
Then, summing \eqref{eq:converse_lemma}  for all $i\in[1:N_1]$ and $j\in[1:N_2]$ provides 
\begin{align} \label{eq:converse_bound1} 
N_2\sum_{i=1}^{N_1}d_{\alpha i}+N_1 \sum_{j=1}^{N_2}d_{\beta j}\leq N_1 N_2.
\end{align}
Obviously, 
\begin{align}
(N_1-N_2)^+ \sum_{i=1}^{N_1}d_{\alpha i}\leq (N_1-N_2)^+ \min(M_1,N_1).\label{eq:converse_bound2}\\
(N_2-N_1)^+ \sum_{j=1}^{N_2}d_{\beta j}\leq (N_2-N_1)^+ \min(M_2,N_2),\label{eq:converse_bound3}
\end{align}
Finally summing \eqref{eq:converse_bound1} to \eqref{eq:converse_bound3} yields 
\begin{align}
\sum_{i=1}^{N_1} d_{\alpha i}+\sum_{j=1}^{N_2} d_{\beta j}\leq\frac{N_1N_2+\min(M_1,N_1)(N_1-N_2)^+ +\min(M_2,N_2)(N_2-N_1)^+}{\max (N_1,N_2)}.
\end{align}
Therefore, $d_{\Sigma}$ is upper bounded by \eqref{eqn:achievable_DoF}, which completes the converse proof.

\begin{figure}[t!]
\begin{center}
\includegraphics[scale=0.7]{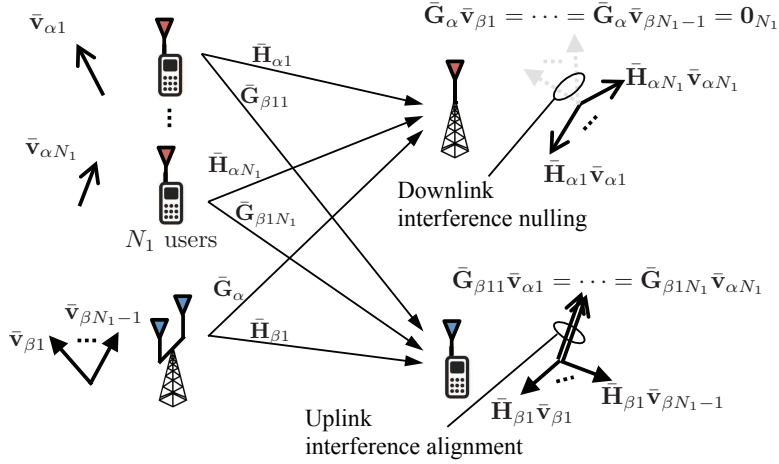}
\end{center}
\vspace{-0.15in}
\caption{$d_{\Sigma}$-achievable transmit beamforming for $M_2=2$, $M_1=N_2=1$.}
\label{figs:simple_case1}
\vspace{-0.1in}
\end{figure}

\subsection{Main Idea for Achievability} \label{subsec:main_idea}
We briefly explain the achievability idea here assuming that $M_2=2$, $M_1=N_2=1$. 
Figure \ref{figs:simple_case1} illustrates how to achieve $d_{\Sigma}=\frac{2N_1-1}{N_1}$ for this case.
Communication takes place via transmit beamforming over a block of $N_1$ time slots. 
Denote 
$\bar{\mathbf{H}}_{\alpha i}=\operatorname{diag}(\mathbf{h}_{\alpha i}[1],\cdots,\mathbf{h}_{\alpha i}[N_1])\in \mathbb{R}^{N_1\times N_1}$,  $\bar{\mathbf{H}}_{\beta 1}=\operatorname{diag}(\mathbf{h}_{\beta 1}[1],\cdots,\mathbf{h}_{\beta 1}[N_1])\in \mathbb{R}^{N_1\times 2N_1}$, $\bar{\mathbf{G}}_{\alpha}=\operatorname{diag}(\mathbf{G}_{\alpha}[1],\cdots,\mathbf{G}_{\alpha}[N_1])\in \mathbb{R}^{N_1\times 2N_1}$, and
$\bar{\mathbf{G}}_{\beta 1i}=\operatorname{diag}(g_{\beta 1i}[1],\cdots,g_{\beta 1i}[N_1])\in \mathbb{R}^{N_1\times N_1}$, where $i\in[1:N_1]$.
As shown in the figure, user $(\alpha,i)$ transmits a single stream via the $N_1\times1$ beamforming vector $\bar{\mathbf{v}}_{\alpha i}$, where $i\in[1:N_1]$.
On the other hand, BS $\beta$  transmits $N_1-1$ streams to its serving user via the $2N_1\times1$ beamforming vectors $\{\bar{\mathbf{v}}_{\beta j}\}_{j\in[1:N_1-1]}$.

Then, we can set linearly independent $\{\bar{\mathbf{v}}_{\alpha i}\}_{i\in[1:N_1]}$ satisfying the uplink IA condition, i.e., $\bar{\mathbf{G}}_{\beta 1i}\bar{\mathbf{v}}_{\alpha i}$ is the same for all $i\in[1:N_1]$.
In particular, for a fixed $\bar{\mathbf{v}}_{\alpha 1}$, set $\bar{\mathbf{v}}_{\alpha i}=\bar{\mathbf{G}}^{-1}_{\beta 1i}\bar{\mathbf{G}}_{\beta 11}\bar{\mathbf{v}}_{\alpha 1}$, where $i\in[2:N_1]$.
We can also set linearly independent $\{\bar{\mathbf{v}}_{\beta j}\}_{j\in[1:N_1-1]}$ satisfying the downlink interference nulling (IN) condition, i.e., $\bar{\mathbf{G}}_{\alpha}\bar{\mathbf{v}}_{\beta j}=\mathbf{0}_{N_1}$ for all $j\in[1:N_1-1]$.
This is possible since the null space for the vector space spanned by the row vectors of $\bar{\mathbf{G}}_{\alpha}$ occupies $N_1$ dimensional subspace in $2N_1$ dimensional space. Therefore set $\{\bar{\mathbf{v}}_{\beta j}\}_{j\in[1:N_1-1]}$ as $N_1-1$ linearly independent vectors in the null space.\footnote{Although $N_1$ linearly independent vectors can satisfy the downlink IN condition, the number of possible streams for successful decoding at user $(\beta,1)$ is given by $N_1-1$ because one dimension is occupied by the inter-cell interference vectors as seen in Fig. \ref{figs:simple_case1}.} 
Hence, BS $\alpha$ is able to decode its $N_1$ intended streams achieving one DoF each since there is no inter-cell interference and $\{\bar{\mathbf{H}}_{\alpha i}\bar{\mathbf{v}}_{\alpha i}\}_{i\in[1:N_1]}$ are linearly independent almost surely.
Similarly, user $(\beta, 1)$ is able to decode its $N_1-1$ intended streams achieving one DoF each since all inter-cell interference vectors are aligned into one dimension and $\{\bar{\mathbf{H}}_{\beta j}\bar{\mathbf{v}}_{\beta j}\}_{j\in[1:N_1-1]}\cup\{\bar{\mathbf{G}}_{\beta 11}\bar{\mathbf{v}}_{\alpha 1}\}$ are linearly independent almost surely.
Finally, from the fact that total $2N_1-1$ streams are delivered over $N_1$ time slots, $d_{\Sigma}=\frac{2N_1-1}{N_1}$ is achievable.

In the following three subsections, we introduce two IA--IN schemes for general $M_1$, $M_2$, $N_1$, and $N_2$ and then derive their achievable sum DoF.
We prove that the maximum achievable sum DoF by the two proposed schemes coincides with $d_{\Sigma}$ in Theorem \ref{thm:achievable_DoF}.
As shown in Fig. \ref{figs:simple_case1}, the first key ingredient follows uplink IA from the users in cell $\alpha$ to the users in cell $\beta$.
Unlike the simple case in Fig. \ref{figs:simple_case1}, asymptotic IA using an arbitrarily large number of time slots is generally needed for simultaneously aligning interference from multiple transmitters at multiple receivers \cite{Viveck1:08}.
The second key ingredient follows downlink IN using $M_2$ antennas from BS $\beta$ to BS $\alpha$ and the users in the same cell.

\subsection{Achievable Sum DoF} \label{subsec:achievable_dof}
We propose two IA--IN schemes generalizing the main idea in Section \ref{subsec:main_idea}.
The first IA--IN scheme applies uplink inter-cell IA and downlink inter-cell and intra-cell IN. 
Specifically, the users in cell $\alpha$ align their interferences at the users in cell $\beta$.
On the other hand, BS $\beta$ nulls out its inter-cell and intra-cell interferences using $M_2$ antennas, each of which is the interference to BS $\alpha$ and the users in cell $\beta$. 
Define $\lambda_1,\lambda_2\in(0,1]$, which are the parameters related to the number of streams for the users in cells $\alpha$ and $\beta$, respectively.
Then the first IA--IN scheme achieves the sum DoF represented by the following optimization problem:
\begin{align} \label{eq:maximization1}
\max_{\substack{
            \lambda_1+\lambda_2\leq 1\\
            N_1\lambda_1\leq M_1\\
            N_1\lambda_1+N_2\lambda_2\leq M_2}}\{N_1\lambda_1+N_2\lambda_2\}.
\end{align}
Here the first constraint, $\lambda_1+\lambda_2\leq 1$, and the second constraint, $N_1\lambda_1\leq M_1$, are needed for successful decoding at the users in cell $\beta$ and BS $\alpha$, respectively.
The last constraint, $ N_1\lambda_1+N_2\lambda_2\leq M_2$ is needed for establishing beamforming vectors for downlink inter-cell and intra-cell IN at BS $\beta$.
The detailed description of  the first IA--IN scheme and the derivation of its achievable sum DoF in \eqref{eq:maximization1} are given in Section \ref{subsec:scheme1}.

Note that the above scheme is not enough to provide the optimal sum DoF for all $M_1$, $M_2$, $N_1$, and $N_2$.
If BS $\alpha$ has a large enough number of antennas (large enough $M_1$), then it is able to decode all intended streams even without downlink inter-cell IN.  
Therefore, for the second IA--IN scheme, downlink beamforming vectors at BS $\beta$ are set only for intra-cell IN, but not for inter-cell IN.
The second IA--IN scheme achieves the sum DoF represented by the following optimization problem:
\begin{align} \label{eq:maximization2}
\max_{\substack{
            \lambda_1+\lambda_2\leq 1\\
            N_1\lambda_1+N_2\lambda_2\leq M_1\\
            N_2\lambda_2\leq M_2}}\{N_1\lambda_1+N_2\lambda_2\}.
\end{align}
Again, the first two constraints are needed for successful decoding at each user in cell $\beta$ and BS $\alpha$ respectively and the last constraint is needed for establishing beamforming vectors at BS $\beta$.
The detailed description of  the second IA--IN scheme and the derivation of its achievable sum DoF in \eqref{eq:maximization2} are given in Section \ref{subsec:scheme2}.

As shown in \eqref{eq:maximization1} and \eqref{eq:maximization2}, there exists a trade-off between the two proposed IA--IN schemes.
The first scheme requires a smaller number of antennas at BS $\alpha$ since the inter-cell interference from BS $\beta$ is zero-forced, which can be verified from the second constraints in \eqref{eq:maximization1} and \eqref{eq:maximization2}.
But at the same time it requires a larger number of antennas at BS $\beta$ since BS $\beta$ have to null out both the inter-cell and intra-cell interferences, which can be verified from the third constraints in \eqref{eq:maximization1} and \eqref{eq:maximization2}.
As a result, the first IA--IN scheme provides a better sum DoF than the second IA--IN scheme if $M_1\leq M_2$, but the second IA--IN scheme provides a better sum DoF for the opposite case, see Table \ref{table:solution} in the Appendix.
More importantly, the following lemma shows that one of the two proposed IA--IN schemes with optimally choosing $\lambda_1$ and $\lambda_2$ achieves $d_{\Sigma}$ for general $M_1$, $M_2$, $N_1$, and $N_2$.

\begin{lemma} \label{lemma:optimal_sol}
Let $d_{\Sigma,1}$ and $d_{\Sigma,2}$ denote the solutions of the two linear programs in \eqref{eq:maximization1} and \eqref{eq:maximization2}, respectively.
Then 
\begin{align} \label{eq:optimal_sol1}
&d_{\Sigma,1}= d_{\Sigma} \mbox{ if }M_1\leq M_2,\nonumber\\
&d_{\Sigma,2}= d_{\Sigma} \mbox{ if }M_2\leq M_1,
\end{align}
where $d_{\Sigma}$ is given by \eqref{eqn:achievable_DoF}.
\end{lemma}
\begin{proof}
We refer to the Appendix for the proof.
\end{proof}

Therefore, Lemma \ref{lemma:optimal_sol} completes the achievability proof of Theorem \ref{thm:achievable_DoF}.
In the next two subsections, we state in details how to achieve \eqref{eq:maximization1} and \eqref{eq:maximization2}.

\begin{remark}[Optimal Scheme for  Cell Coordination]
For the cell coordination problem, e.g., stated in Remark \ref{re:uplink_downlink} and Section \ref{subsec:dof_gain}, only one of the two proposed IA--IN schemes is enough to maximize the sum DoF achievable by uplink--downlink operation. 
In particular, we can attain the maximum sum DoF achievable by uplink--downlink operation using the first IA--IN scheme by operating the cell having more BS antennas as downlink (and the other cell as uplink). \hfill$\lozenge$
\end{remark}

\subsection{Uplink Inter-Cell IA and Downlink Inter-Cell and Intra-Cell IN} \label{subsec:scheme1}
To prove that \eqref{eq:maximization1} is achievable, we state the first IA--IN scheme, which applies uplink inter-cell IA and downlink inter-cell and intra-cell IN.

From now on, $\lambda_1,\lambda_2\in(0,1]$ are assumed to be set such that they satisfy the three constraints in \eqref{eq:maximization1}.
Define $\mathcal{S}_T =[0:T-1]^{N_1N_2}$. We first divide $W_{\alpha i}$, $i\in[1:N_1]$, into $T^{N_1N_2}$ submessages $\big\{W^{(\mathbf{s})}_{\alpha i}\big\}_{\mathbf{s}\in\mathcal{S}_T}$.
Let $\big[c^{(\mathbf{s})}_{\alpha i}[1],\cdots,c^{(\mathbf{s})}_{\alpha i}[n]\big]$ denote a length-$n$ codeword of Gaussian codebook generated i.i.d. from $\mathcal{N}(0,P)$, that is associated with $W^{(\mathbf{s})}_{\alpha i}$.
Similarly, divide $W_{\beta j}$, $j\in[1:N_2]$, into $\frac{\lambda_2}{\lambda_1}T^{N_1N_2}$ submessages $\big\{W^{(k)}_{\beta j}\big\}_{k\in[1:\frac{\lambda_2}{\lambda_1}T^{N_1N_2}]}$.
Let $\big[c^{(k)}_{\beta j}[1],\cdots,c^{(k)}_{\beta j}[n]\big]$ denote a length-$n$ codeword of Gaussian codebook generated i.i.d. from $\mathcal{N}(0,P)$, that is associated with $W^{(k)}_{\beta j}$.

Let $d=\frac{1}{\lambda_1}(T+1)^{N_1 N_2}$. Communication will take place over a block of $nd$ time slots. 
Each of the codewords defined above will be transmitted via a length-$d$ time-extended beamforming vector.
For easy explanation, denote the length-$d$ time-extended inputs and outputs as
\begin{align} \label{eq:time_extended_inout}
\bar{\mathbf{x}}_{\alpha i}[m]&=\left[x_{\alpha i}[(m-1)d+1],\cdots,  x_{\alpha i}[md]\right]^{\dagger}\in\mathbb{R}^{d\times 1},\nonumber\\
\bar{\mathbf{x}}_{\beta}[m]&=\left[\mathbf{x}_{\beta}[(m-1)d+1],\cdots,  \mathbf{x}_{\beta}[md]\right]^{\dagger}\in\mathbb{R}^{M_2d\times 1},\nonumber\\
\bar{\mathbf{y}}_{\alpha}[m]&=\left[\mathbf{y}_{\alpha}[(m-1)d+1],\cdots,  \mathbf{y}_{\alpha}[md]\right]^{\dagger}\in\mathbb{R}^{M_1d\times 1},\nonumber\\
\bar{\mathbf{y}}_{\beta j}[m]&=\left[y_{\beta j}[(m-1)d+1],\cdots,  y_{\beta j}[md]\right]^{\dagger}\in\mathbb{R}^{d\times 1},
\end{align} 
where $m\in[1:n]$.
Then from \eqref{eq:in_out_1} and \eqref{eq:in_out_2}
\begin{align} \label{eq:time_extended1}
\bar{\mathbf{y}}_{\alpha}[m]&=\sum_{i=1}^{N_1}\bar{\mathbf{H}}_{\alpha i}[m]\bar{\mathbf{x}}_{\alpha i}[m]+\bar{\mathbf{G}}_{\alpha}[m]\bar{\mathbf{x}}_{\beta}[m]+\bar{\mathbf{z}}_{\alpha}[m],\nonumber\\
\bar{\mathbf{y}}_{\beta j}[m]&=\bar{\mathbf{H}}_{\beta j}[m]\bar{\mathbf{x}}_{\beta}[m]+\sum_{i=1}^{N_1}\bar{\mathbf{G}}_{\beta ji}[m]\bar{\mathbf{x}}_{\alpha i}[m]+\bar{\mathbf{z}}_{\beta j}[m],
\end{align}
where
\begin{align} \label{eq:time_extended_ch}
\bar{\mathbf{H}}_{\alpha i}[m]&=\operatorname{diag}(\mathbf{h}_{\alpha i}[(m-1)d+1],\cdots,\mathbf{h}_{\alpha i}[md])\in\mathbb{R}^{M_1d\times d},\nonumber\\
\bar{\mathbf{H}}_{\beta j}[m]&=\operatorname{diag}(\mathbf{h}_{\beta j}[(m-1)d+1],\cdots,\mathbf{h}_{\beta j}[md])\in\mathbb{R}^{d\times M_2d},\nonumber\\
\bar{\mathbf{G}}_{\alpha }[m]&=\operatorname{diag}(\mathbf{G}_{\alpha }[(m-1)d+1],\cdots,\mathbf{G}_{\alpha }[md])\in\mathbb{R}^{M_1d\times M_2d},\nonumber\\
\bar{\mathbf{G}}_{\beta ji }[m]&=\operatorname{diag}(g_{\beta ji}[(m-1)d+1],\cdots,g_{\beta ji}[md])\in\mathbb{R}^{d\times d}
\end{align}
and 
\begin{align}
\bar{\mathbf{z}}_{\alpha}[m]&=\left[\mathbf{z}_{\alpha}[(m-1)d+1],\cdots,  \mathbf{z}_{\alpha}[md]\right]^{\dagger}\in\mathbb{R}^{M_1 d\times 1},\nonumber\\
\bar{\mathbf{z}}_{\beta j}[m]&=\left[z_{\beta j}[(m-1)d+1],\cdots,  z_{\beta j}[md]\right]^{\dagger}\in\mathbb{R}^{d\times 1}.
\end{align}

\subsubsection{Transmit beamforming for IA and IN} \label{subsec:proposed_IA1}

\begin{figure}[t!]
\begin{center}
\includegraphics[scale=0.7]{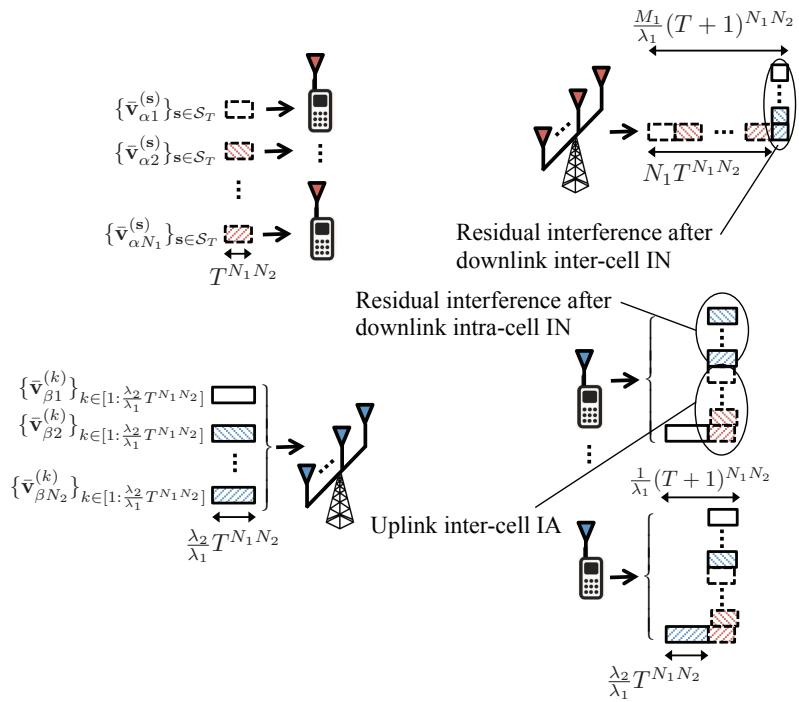}
\end{center}
\vspace{-0.15in}
\caption{Uplink inter-cell IA and downlink inter-cell and intra-cell IN, where for convenience we assume $\lambda_1\leq \lambda_2$ in the figure.}
\label{fig:scheme_i_bc_mac}
\vspace{-0.1in}
\end{figure}

For $m\in[1:n]$ and $\mathbf{s}\in \mathcal{S}_T$, $c^{(\mathbf{s})}_{\alpha i}[m]$ is transmitted via a length-$d$ time-extended beamforming vector $\bar{\mathbf{v}}^{(\mathbf{s})}_{\alpha i}[m]\in \mathbb{R}^{d\times 1}$.
Similarly, for $m\in[1:n]$ and $k\in[1:\frac{\lambda_2}{\lambda_1}T^{N_1N_2}]$, $c^{(k)}_{\beta j}[m]$ is transmitted via a length-$d$ time-extended beamforming vector $\bar{\mathbf{v}}^{(k)}_{\beta j}[m]\in \mathbb{R}^{M_2d\times 1}$.
That is, user $(\alpha, i)$ transmits
\begin{equation} \label{eq:input_vec1}
\mathbf{x}_{\alpha i}[m]=\gamma\sum_{\mathbf{s}\in \mathcal{S}_T}\bar{\mathbf{v}}^{(\mathbf{s})}_{\alpha i}[m]c^{(\mathbf{s})}_{\alpha i}[m],
\end{equation}
and BS $\beta$ transmits
\begin{equation} \label{eq:input_vec2}
\mathbf{x}_{\beta}[m]=\gamma\sum_{j=1}^{N_2}\sum_{k=1}^{\frac{\lambda_2}{\lambda_1}T^{N_1 N_2}}\bar{\mathbf{v}}^{(k)}_{\beta j}[m]c^{(k)}_{\beta j}[m],
\end{equation}
where $\gamma>0$ is chosen to satisfy the average power $P$.
Figure \ref{fig:scheme_i_bc_mac} illustrates how to construct these length-$d$ time-extended  beamforming vectors for  uplink inter-cell IA and downlink inter-cell and intra-cell IN. The detailed construction of such beamforming vectors is explained in the following. Since the overall construction is identical for all $m\in[1:n]$, we assume $m=1$ and omit the index $m$ from now on.
\\\\
{\bf Uplink inter-cell IA:}
\\
To align inter-cell interference from $N_1$ users in cell $\alpha$ to $N_2$ users in cell $\beta$, asymptotic signal space alignment is needed, originally proposed in \cite{Viveck1:08}.
In this paper, we adopt a recent framework developed in \cite{Shomorony:13} for asymptotic signal space alignment.
For $\mathbf{s}=[s_{11},s_{12},\cdots,s_{N_2 N_1}]\in \mathcal{S}_T$, define
\begin{equation} \label{eq:direction}
v^{(\mathbf{s})}[t]=\prod_{1\leq i\leq N_1,1\leq j\leq N_2}g_{\beta ji}[t]^{s_{ji}}
\end{equation}
for $t\in[1:d]$ and $\bar{\mathbf{v}}^{(\mathbf{s})}=[v^{(\mathbf{s})}[1],\cdots,v^{(\mathbf{s})}[d]]^{\dagger}$.
Set 
\begin{align} \label{eq:beam1}
\bar{\mathbf{v}}^{(\mathbf{s})}_{\alpha i}=\bar{\mathbf{v}}^{(\mathbf{s})}
\end{align}
for all $i\in[1:N_1]$ and $\mathbf{s}\in\mathcal{S}_T$.
The following lemma  shows that the beamforming vectors defined in \eqref{eq:direction} and \eqref{eq:beam1}  guarantee asymptotic uplink inter-cell IA at the users in cell $\beta$. 
\begin{lemma} \label{lemma:dim}
The signal space spanned by $\{\bar{\mathbf{G}}_{\beta ji}\bar{\mathbf{v}}^{(\mathbf{s})}_{\alpha i}\}_{i\in[1:N_1],j\in[1:N_2],\mathbf{s}\in\mathcal{S}_T}$ occupies at least $T^{N_1N_2}$ dimensional subspace and at most $(T+1)^{N_1 N_2}$ dimensional subspace in $\frac{1}{\lambda_1}(T+1)^{N_1N_2}$ dimensional space almost surely. 
\end{lemma}
\begin{proof}
From the fact that $\{\bar{\mathbf{v}}^{(\mathbf{s})}\}_{\mathbf{s}\in\mathcal{S}_T}$ is a set of $T^{N_1N_2}$ linearly independent vectors almost surely \cite{Shomorony:13}, $\operatorname{span}\big(\{\bar{\mathbf{G}}_{\beta ji}\bar{\mathbf{v}}^{(\mathbf{s})}_{\alpha i}\}_{i\in[1:N_1],j\in[1:N_2],\mathbf{s}\in\mathcal{S}_T}\big)$ occupies at least $T^{N_1N_2}$ dimensional subspace almost surely. 

Now consider the upper bound. For all $i\in[1:N_1]$, $j\in[1:N_2]$, and $\mathbf{s}\in\mathcal{S}_T$, 
\begin{equation}
\bar{\mathbf{G}}_{\beta ji}\bar{\mathbf{v}}^{(\mathbf{s})}\in \{\bar{\mathbf{v}}^{(\mathbf{s}')}\}_{\mathbf{s}'\in\mathcal{S}_{T+1}}
\end{equation}
showing that $\operatorname{span}\big(\{\bar{\mathbf{G}}_{\beta ji}\bar{\mathbf{v}}^{(\mathbf{s})}_{\alpha i}\}_{i\in[1:N_1],j\in[1:N_2],\mathbf{s}\in\mathcal{S}_T}\big)$ occupies at most $(T+1)^{N_1N_2}$ dimensional subspace since the cardinality of $\mathcal{S}_{T+1}$ is given by $(T+1)^{N_1N_2}$.
Therefore, Lemma \ref{lemma:dim} holds.
\end{proof}
\vspace{0.2in}
{\bf Downlink inter-cell and intra-cell IN:}
\\
From \eqref{eq:time_extended1}, \eqref{eq:input_vec1}, and \eqref{eq:input_vec2}.
\begin{align}
\bar{\mathbf{y}}_{\alpha}&=\gamma\sum_{i=1}^{N_1}\sum_{\mathbf{s}\in \mathcal{S}_T}\bar{\mathbf{H}}_{\alpha i}\bar{\mathbf{v}}^{(\mathbf{s})}_{\alpha i}c^{(\mathbf{s})}_{\alpha i}+\gamma\sum_{j=1}^{N_2}\sum_{k=1}^{\frac{\lambda_2}{\lambda_1}T^{N_1 N_2}}\bar{\mathbf{G}}_{\alpha}\bar{\mathbf{v}}^{(k)}_{\beta j}c^{(k)}_{\beta j}+\bar{\mathbf{z}}_{\alpha},\nonumber\\
\bar{\mathbf{y}}_{\beta j}&=\gamma\sum_{j=1}^{N_2}\sum_{k=1}^{\frac{\lambda_2}{\lambda_1}T^{N_1 N_2}}\bar{\mathbf{H}}_{\beta j}\bar{\mathbf{v}}^{(k)}_{\beta j}c^{(k)}_{\beta j}+\gamma\sum_{i=1}^{N_1}\sum_{\mathbf{s}\in \mathcal{S}_T}\bar{\mathbf{G}}_{\beta ji}\bar{\mathbf{v}}^{(\mathbf{s})}_{\alpha i}c^{(\mathbf{s})}_{\alpha i}+\bar{\mathbf{z}}_{\beta j}.
\end{align}
Hence, in order to null out inter-cell interference by zero-forcing at BS $\alpha$,
\begin{align} \label{eq:cond1}
\bar{\mathbf{G}}_{\alpha}\bar{\mathbf{v}}^{(k)}_{\beta j}\perp\operatorname{span}\left(\{\bar{\mathbf{H}}_{\alpha i'}\bar{\mathbf{v}}^{(\mathbf{s})}_{\alpha i'}\}_{i'\in[1:N_1],\mathbf{s}\in\mathcal{S}_T}\right)
\end{align}
for all $j\in[1:N_2]$ and $k\in[1:\frac{\lambda_2}{\lambda_1}T^{N_1N_2}]$.

In order to null out intra-cell interference, we first define $\frac{\lambda_2}{\lambda_1}T^{N_1N_2}$ dimensional subspace in $\frac{1}{\lambda_1}(T+1)^{N_1 N_2}$ dimensional space represented by $\operatorname{span}\left(\{\bar{\mathbf{w}}_{k'}\}_{k'\in[1:\frac{\lambda_2}{\lambda_1}T^{N_1N_2}]}\right)$, which will be used for the signal space of the intended submessages at the users in cell $\beta$.
From Lemma \ref{lemma:dim}, $\operatorname{span}\left(\{\bar{\mathbf{G}}_{\beta j'i'}\bar{\mathbf{v}}^{(\mathbf{s})}_{\alpha i'}\}_{i'\in[1:N_1],j'\in[1:N_2],\mathbf{s}\in\mathcal{S}_T}\right)$ occupies at most $(T+1)^{N_1 N_2}$ dimensions almost surely, which means the null space of  $\operatorname{span}\left(\{\bar{\mathbf{G}}_{\beta j'i'}\bar{\mathbf{v}}^{(\mathbf{s})}_{\alpha i'}\}_{i'\in[1:N_1],j'\in[1:N_2],\mathbf{s}\in\mathcal{S}_T}\right)$ occupies at least $\frac{1}{\lambda_1}(T+1)^{N_1 N_2}-(T+1)^{N_1 N_2}$ dimensions almost surely.  
Hence we set $\{\bar{\mathbf{w}}_{k'}\}_{k'\in[1:\frac{\lambda_2}{\lambda_1}T^{N_1N_2}]}$ as a subset of  $\frac{\lambda_2}{\lambda_1}T^{N_1N_2}$ basis consisting of the null space of $\operatorname{span}\left(\{\bar{\mathbf{G}}_{\beta j'i'}\bar{\mathbf{v}}^{(\mathbf{s})}_{\alpha i'}\}_{i'\in[1:N_1],j'\in[1:N_2],\mathbf{s}\in\mathcal{S}_T}\right)$.
This is possible because 
\begin{align}
\frac{1}{\lambda_1}(T+1)^{N_1 N_2}-(T+1)^{N_1 N_2} \geq \frac{\lambda_2}{\lambda_1}T^{N_1 N_2},
\end{align}
where the inequality follows since $\lambda_1+\lambda_2\leq 1$.
Therefore, for the intra-cell IN by zero-forcing at the users in cell $\beta$, 
\begin{align} \label{eq:cond2}
\bar{\mathbf{H}}_{\beta i}\bar{\mathbf{v}}^{(k)}_{\beta j}\perp \operatorname{span}\left(\{\bar{\mathbf{w}}_{k'}\}_{k'\in[1:\frac{\lambda_2}{\lambda_1}T^{N_1N_2}]}\right)
\end{align}
should be satisfied for all $i,j\in[1:N_2]$, $i\neq j$, and $k\in[1:\frac{\lambda_2}{\lambda_1}T^{N_1N_2}]$.

As a consequence, from \eqref{eq:cond1} and \eqref{eq:cond2}, $\bar{\mathbf{v}}^{(k)}_{\beta j}$ should be orthogonal with the following  vectors:
\begin{align} \label{eq:cond3}
&\{\bar{\mathbf{G}}_{\alpha}^{\dagger} \bar{\mathbf{H}}_{\alpha i'}\bar{\mathbf{v}}^{(\mathbf{s})}_{\alpha i'}\}_{i'\in[1:N_1],\mathbf{s}\in\mathcal{S}_T},\nonumber\\
&\{\bar{\mathbf{H}}^{\dagger}_{\beta i'}\bar{\mathbf{w}}_{k'}\}_{i'\in[1:N_2],i'\neq j,k'\in[1:\frac{\lambda_2}{\lambda_1}T^{N_1N_2}]}.
\end{align}
Since there are total $(N_1+\frac{\lambda_2}{\lambda_1}(N_2-1))T^{N_1 N_2}$ vectors in \eqref{eq:cond3} and $\bar{\mathbf{v}}^{(k)}_{\beta j}$ has $\frac{M_2}{\lambda_1} (T+1)^{N_1 N_2}$ elements,
we can set linearly independent $\{\bar{\mathbf{v}}^{(k)}_{\beta j}\}_{k\in[1:\frac{\lambda_2}{\lambda_1}T^{N_1N_2}]}$ orthogonal with the vectors in  \eqref{eq:cond3} for all $j\in[1:N_2]$ if
\begin{align} \label{eq:condition3}
\frac{M_2}{\lambda_1}(T+1)^{N_1 N_2}-(N_1+\frac{\lambda_2}{\lambda_1}(N_2-1))T^{N_1 N_2}>\frac{\lambda_2}{\lambda_1}T^{N_1N_2},
\end{align}
which is satisfied from the assumption that
\begin{align}  \label{eq:condition33}
 N_1\lambda_1 +N_2\lambda_2  \leq M_2.
\end{align}
In conclusion, $\{\bar{\mathbf{v}}^{(k)}_{\beta j}\}_{j\in[1:N_2],k\in[1:\frac{\lambda_2}{\lambda_1}T^{N_1N_2}]}$ can be set to satisfy the downlink inter-cell and intra-cell IN conditions almost surely.

\subsubsection{Zero-forcing decoding}
Each submessage will be decoded by zero-forcing. we first introduce the following properties:
\begin{itemize}
\item[(A)] $\bar{\mathbf{v}}^{(\mathbf{s})}_{\alpha i}$ is a function of $\{\bar{\mathbf{G}}_{\beta j'i'}\}_{i'\in[1:N_1],j'\in[1:N_2]}$ (see \eqref{eq:direction} and \eqref{eq:beam1})
\item[(B)] $\bar{\mathbf{v}}^{(k)}_{\beta j}$ is a function of $\{\bar{\mathbf{H}}_{\alpha i'}\}_{i'\in[1:N_1]}$, $\{\bar{\mathbf{H}}_{\beta j'}\}_{j'\in[1:N_2],j'\neq j}$, $\bar{\mathbf{G}}_{\alpha}$, and $\{\bar{\mathbf{G}}_{\beta j'i'}\}_{i'\in[1:N_1],j'\in[1:N_2]}$ (see \eqref{eq:cond3} and Property (A)),
\end{itemize}
Based on the above properties, we  prove that one DoF is achievable for each submessage.
\\\\
{\bf Decoding at BS $\alpha$:}
\\
Since $\{\bar{\mathbf{v}}^{(k)}_{\beta j}\}_{j\in[1:N_2],k\in[1:\frac{\lambda_2}{\lambda_1}T^{N_1N_2}]}$ is set to satisfy the inter-cell IN condition in \eqref{eq:cond1}, inter-cell interference will disappear after zero-forcing at BS $\alpha$.
Hence, in order to achieve one DoF for each submessage, $\{\bar{\mathbf{H}}_{\alpha i'}\bar{\mathbf{v}}^{(\mathbf{s})}_{\alpha i'}\}_{i'\in[1:N_1],\mathbf{s}\in\mathcal{S}_T}$ should be a set of linearly independent vectors.
Note that $\{\mathbf{v}^{(\mathbf{s})}_{\alpha i'}\}_{\mathbf{s}\in\mathcal{S}_{T}}$ is a set of linearly independent vectors almost surely \cite{Shomorony:13}.
Furthermore, from Property (A), $\bar{\mathbf{H}}_{\alpha i'}\bar{\mathbf{v}}^{(\mathbf{s})}_{\alpha i'}$ is a random projection of $\bar{\mathbf{v}}^{(\mathbf{s})}_{\alpha i'}$ into $M_1 d$ dimensional space ($\bar{\mathbf{v}}^{(\mathbf{s})}_{\alpha i'}$ is set independent of $\bar{\mathbf{H}}_{\alpha i'}$).
Therefore, $\{\bar{\mathbf{H}}_{\alpha i'}\bar{\mathbf{v}}^{(\mathbf{s})}_{\alpha i'}\}_{i'\in[1:N_1],\mathbf{s}\in\mathcal{S}_T}$ is a set of linearly independent vectors almost surely if
\begin{align}
N_1T^{N_1 N_2}\leq \frac{M_1}{\lambda_1}(T+1)^{N_1N_2},
\end{align}
which is satisfied from the assumption that 
\begin{align} \label{eq:dec_con1}
N_1\lambda_1\leq M_1.
\end{align}
In conclusion,  each submessage intended to BS $\alpha$ can be decoded by achieving one DoF almost surely.
\\\\
{\bf Decoding at the users in cell $\beta$:}
\\
Consider the decoding at user $(\beta,j)$, where $j\in[1:N_2]$.
Since $\{\bar{\mathbf{v}}^{(k)}_{\beta j'}\}_{j'\in[1:N_2],k\in[1:\frac{\lambda_2}{\lambda_1}T^{N_1N_2}]}$ is set to satisfy the intra-cell IN condition in \eqref{eq:cond2}, intra-cell interference will disappear after zero-forcing.
Hence, in order to achieve one DoF for each submessage, 
$\{\bar{\mathbf{H}}_{\beta j}\bar{\mathbf{v}}^{(k)}_{\beta j}\}_{k\in[1:\frac{\lambda_2}{\lambda_1}T^{N_1 N_2}]}$ should be a set of linearly independent vectors and 
\begin{align} \label{eq:independency}
\bar{\mathbf{H}}_{\beta j}\bar{\mathbf{v}}^{(k)}_{\beta j}\notin\operatorname{span}\left(\{\bar{\mathbf{G}}_{\beta ji'}\bar{\mathbf{v}}^{(\mathbf{s})}_{\alpha i'}\}_{i'\in[1:N_1],\mathbf{s}\in\mathcal{S}_T}\right)
\end{align}
should be satisfied for all $k\in[1:\frac{\lambda_2}{\lambda_1}T^{N_1 N_2}]$.

First consider the linearly independent condition.
From Property (B), $\bar{\mathbf{H}}_{\beta j}\bar{\mathbf{v}}^{(k)}_{\beta j}$ is a random projection of $\bar{\mathbf{v}}^{(k)}_{\beta j}$ into $\frac{1}{\lambda_1}(T+1)^{N_1 N_2}$ dimensional space ($\bar{\mathbf{v}}^{(k)}_{\beta j}$ is set independent of $\bar{\mathbf{H}}_{\beta j}$).
Hence, $\{\bar{\mathbf{H}}_{\beta j}\bar{\mathbf{v}}^{(k)}_{\beta j}\}_{k\in[1:\frac{\lambda_2}{\lambda_1}T^{N_1 N_2}]}$ is a set of linearly independent vectors almost surely since $\frac{\lambda_2}{\lambda_1}T^{N_1 N_2}\leq \frac{1}{\lambda_1}(T+1)^{N_1N_2}$.

Now consider the condition in \eqref{eq:independency}.
Lemma \ref{lemma:dim} shows that $\operatorname{span}\left(\{\bar{\mathbf{G}}_{\beta ji'}\bar{\mathbf{v}}^{(\mathbf{s})}_{\alpha i'}\}_{i'\in[1:N_1],\mathbf{s}\in\mathcal{S}_T}\right)$ occupies at most $(T+1)^{N_1 N_2}$ dimensions due to the uplink inter-cell IA.
From Property (B), $\bar{\mathbf{H}}_{\beta j}\bar{\mathbf{v}}^{(k)}_{\beta j}$ is a random projection of $\bar{\mathbf{v}}^{(k)}_{\beta j}$ into $d$ dimensional space ( $\bar{\mathbf{v}}^{(k)}_{\beta j}$ is set independent of $\bar{\mathbf{H}}_{\beta j}$) and $\operatorname{span}\left(\{\bar{\mathbf{G}}_{\beta ji}\bar{\mathbf{v}}^{(\mathbf{s})}_{\alpha i}\}_{i\in[1:N_1],\mathbf{s}\in\mathcal{S}_T}\right)$ is independent of $\bar{\mathbf{H}}_{\beta j}$. 
Therefore \eqref{eq:independency} is satisfied almost surely if 
\begin{align} \label{eq:dec_con22}
\frac{\lambda_2}{\lambda_1}T^{N_1 N_2}+(T+1)^{N_1 N_2}\leq \frac{1}{\lambda_1}(T+1)^{N_1 N_2},
\end{align}
which is satisfied from the assumption that 
\begin{align}  \label{eq:dec_con2}
\lambda_1+\lambda_2\leq 1.
\end{align}
 In conclusion, each submessage intended to the users in the second cell can be decoded by achieving one DoF almost surely.

\subsubsection{Achievable Sum DoF}
From the facts that each submessage is delivered via a length-$n$ codeword and total $(N_1+\frac{\lambda_2}{\lambda_1}N_2)T^{N_1 N_2}$ submessages are delivered during $nd=n\frac{1}{\lambda_1}(T+1)^{N_1 N_2}$ time slots, the sum DoF 
\begin{equation} \label{eq:achievable_DoF1}
\frac{(N_1+\frac{\lambda_2}{\lambda_1}N_2)T^{N_1 N_2}}{\frac{1}{\lambda_1}(T+1)^{N_1 N_2}}
\end{equation}
is achievable under the three constraints in \eqref{eq:condition33}, \eqref{eq:dec_con1}, and \eqref{eq:dec_con2}. 
Finally,  since \eqref{eq:achievable_DoF1} converges to $N_1 \lambda_1+N_2 \lambda_2$ as $T$ increases, the sum DoF in \eqref{eq:maximization1} is achievable.

\subsection{Uplink Inter-Cell IA and Downlink Intra-Cell IN} \label{subsec:scheme2}
In this subsection, we prove that \eqref{eq:maximization2} is achievable.
Assume that $\lambda_1,\lambda_2\in(0,1]$ are set such that they satisfy the three constraints in \eqref{eq:maximization2}.
The second IA--IN scheme briefly explained in Section \ref{subsec:achievable_dof} is a simple modification of the first IA--IN scheme.
The overall transmission based on the length-$d$ time-extended transmit beamforming is the same as in Section \ref {subsec:scheme1}.
The uplink inter-cell IA is the same as in Section \ref {subsec:scheme1}.
For downlink beamforming at BS $\beta$, on the other hand, $\{\bar{\mathbf{v}}^{(k)}_{\beta j}\}_{k\in[1:\frac{\lambda_2}{\lambda_1}T^{N_1N_2}]}$ is set only for the intra-cell IN, but not for inter-cell IN.
That is,  \eqref{eq:cond2} should be satisfied  for all $i,j\in[1:N_2]$, $i\neq j$, and $k\in[1:\frac{\lambda_2}{\lambda_1}T^{N_1N_2}]$, where $\{\bar{\mathbf{w}}_{k'}\}_{k'\in[1:\frac{\lambda_2}{\lambda_1}T^{N_1N_2}]}$ is defined in Section \ref{subsec:scheme1}.
Therefore, $\bar{\mathbf{v}}^{(k)}_{\beta j}$ should be orthogonal with the following  vectors:
\begin{align} \label{eq:cond4}
\{\bar{\mathbf{H}}^{\dagger}_{\beta i'}\bar{\mathbf{w}}_{k'}\}_{i'\in[1:N_2],i'\neq j,k'\in[1:\frac{\lambda_2}{\lambda_1}T^{N_1N_2}]}.
\end{align}
Since there are total $\frac{\lambda_2}{\lambda_1}(N_2-1)T^{N_1 N_2}$ vectors in \eqref{eq:cond4} and $\bar{\mathbf{v}}^{(k)}_{\beta j}$ has $\frac{M_2}{\lambda_1} (T+1)^{N_1 N_2}$ elements,
we can set linearly independent $\{\bar{\mathbf{v}}^{(k)}_{\beta j}\}_{k\in[1:\frac{\lambda_2}{\lambda_1}T^{N_1N_2}]}$ orthogonal with the vectors in  \eqref{eq:cond4} for all $j\in[1:N_2]$ if
\begin{align} \label{eq:condition5}
\frac{M_2}{\lambda_1}(T+1)^{N_1 N_2}-\frac{\lambda_2}{\lambda_1}(N_2-1)T^{N_1 N_2}>\frac{\lambda_2}{\lambda_1}T^{N_1N_2},
\end{align}
which is satisfied from the assumption that $N_2\lambda_2  \leq M_2$.

Now consider the decoding procedure.
Even though inter-cell interference from BS $\beta$ is not zero-forced, BS $\alpha$ is able to decode all the intended submessages by zero-forcing if the number of dimensions occupied by all signal and interference vectors is less than or equal to $M_1 d$, i.e.,
\begin{align}
N_1 T^{N_1 N_2}+\frac{\lambda_2}{\lambda_1}N_2 T^{N_1 N_2}\leq \frac{M_1}{\lambda_1}(T+1)^{N_1 N_2},
\end{align}
which is satisfied from the assumption that $\lambda_1N_1+\lambda_2 N_2\leq M_1$.
Lastly, the condition for successful decoding at each user in cell $\beta$ is the same as  in \eqref{eq:dec_con22},
which is satisfied from the assumption that $\lambda_1+\lambda_2\leq 1$.
Therefore, the second IA--IN scheme achieves the sum DoF in \eqref{eq:maximization2}.

\section{Discussion} \label{sec:discussion}
In this section, we discuss about the cell coordination problem figuring out the DoF gain achievable by uplink--downlink operation in more details in Sections \ref{subsec:dof_gain} and \ref{subsec:dof_hetnet} and also propose a simple IA scheme exploiting delayed CSI at transmitters (CSIT) in Section \ref{subsec:delayed_csit}.

\subsection{DoF Gain From Uplink--Downlink Operation} \label{subsec:dof_gain}

In Remark 1 of Section \ref{sec:main_result}, we have briefly explained the DoF gain achievable by uplink--downlink operation compared to the conventional uplink or downlink operation.
In this subsection, we consider the \emph{cell coordination problem} in more details for a general four-parameter space $(M_1,M_2,N_1,N_2)$. Specifically, the first cell consists of the BS with $M_1$ antennas and $N_1$ users and the second cell consists of the BS with $M_2$ antennas and $N_2$ users. The operation mode of each cell can be chosen to maximize the sum DoF. 

Unfortunately, the sum DoF of the two-cell multiantenna IBC (or IMAC) is not completely characterized for a general $(M_1,M_2,N_1,N_2)$. It was shown in \cite{Park:12} that, for $\max(M_1,M_2)\geq \min(N_1,N_2)$, the sum DoF is given by 
\begin{equation}
\min\left\{N_1+N_2,\max(M_1,N_2),\max(M_2,N_1)\right\},
\end{equation}
which corresponds to the regime that zero-forcing is optimal.
For $\max(M_1,M_2)\leq \min(N_1,N_2)$, on the other hand, zero-forcing is not optimal in general and the sum DoF has been characterized only for  the symmetric case where $M_1=M_2:=M$ and $N_1=N_2:=N$.
Specifically, the sum DoF is given by $\frac{2MN}{M+N}$ if $M\leq N$ \cite{Sridharan:13}, which is achievable by treating each BS antenna as a separate user and then applying asymptotic IA proposed in \cite{Viveck2:09}.

To figure out the DoF gain from uplink--downlink operation over a four-parameter space $(M_1,M_2,N_1,N_2)$, for $\Lambda\in\mathbb{Z}_+$,
we define
\begin{align}
\delta_{\sf{gain}}(\Lambda):=\frac{\sum_{i,j,k,l\in[1:\Lambda]}1_{d_{\Sigma}(i,j,k,l)>d_{\sf upper}(i,j,k,l)}}{\Lambda^4},
\end{align}
where 
\begin{align} 
&d_{\Sigma}(i,j,k,l)\nonumber\\
&=\min\Bigg\{\frac{kl+\min(i,k)(k-l)^+ +\min(j,l)(l-k)^+}{\max (k,l)},i+l,j+k,\max(i,j),\max(k,l)\Bigg\},
\end{align}
\begin{align} 
d_{\sf upper}(i,j,k,l)=\min\{i+j,k+l,\max(i,l),\max(j,k)\},
\end{align}
and $1_{(\cdot)}$ denotes the indicator function.
Note that $d_{\Sigma}(i,j,k,l)$ is given from Theorem 1, which is the sum DoF obtained by uplink--downlink operation, and  $d_{\sf upper}(i,j,k,l)$ is an upper bound on the sum DoF obtained by the conventional uplink or downlink operation \cite{Jafar:07}.
Hence, from the definition of $\delta_{\sf gain}(\Lambda)$, uplink--downlink operation is beneficial for improving the sum DoF at least $\delta_{\sf gain}(\Lambda)$ fraction of the entire four-parameter space $(M_1,M_2,N_1,N_2)$.
Table \ref{table:gain_fraction} states $\delta_{\sf gain}(\Lambda)$ with respect to $\Lambda$.
As the space size $\Lambda$ increases, the fraction of subspace showing the DoF gain from uplink--downlink operation increases. For instance, uplink--downlink operation can improve the sum DoF more than 30 percent of the entire space when $\Lambda=32$.

\begin{table}[t!] \caption{A fraction of the four-parameter space showing the DoF gain from uplink--downlink operation.}
\label{table:gain_fraction}
\begin{center}
\begin{tabular}{c|c|c|c|c|c|c}
\hline 
$\Lambda$ & 2 & 4 & 8 & 16 & 32 & 64 \\ 
\hline 
$\delta_{\sf{gain}}(\Lambda)$ & 0.1250 & 0.2031 & 0.2598 & 0.2942 & 0.3131 & 0.3231 \\ 
\hline 
\end{tabular} 
\end{center}
\end{table}%

From Table \ref{table:solution}, which will be explained in the Appendix, we can see that except the regimes 5, 6, 9, 10, 15, 16, 19, and 20, single-cell operation achieves $d_{\Sigma}$.
Hence the same sum DoF is also achievable by either uplink or downlink operation (with single-cell operation), meaning that uplink--downlink operation cannot improve the sum DoF except for the regimes 5, 6, 9, 10, 15, 16, 19, and 20.
The numerical result in Table \ref{table:gain_fraction} demonstrates that uplink--downlink operation strictly improves the sum DoF for most of the cases in regimes 5, 6, 9, 10, 15, 16, 19, and 20, which is $8$ regimes out of $24$ regimes.

\begin{figure}[t!]
\begin{center}
\includegraphics[scale=0.7]{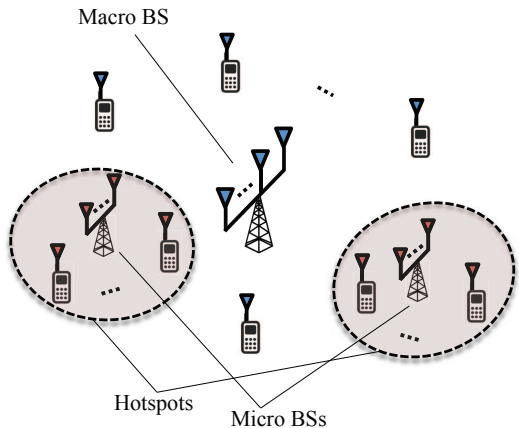}
\end{center}
\vspace{-0.15in}
\caption{Heterogeneous cellular networks having hotspots in which the users in each hotspot are served from the micro BS in the same hotspot.}
\label{figs:hotspots}
\vspace{-0.1in}
\end{figure}

\subsection{DoF of Heterogeneous Cellular Networks} \label{subsec:dof_hetnet}
Recently, heterogeneous cellular networks called ``HetNets'' have been actively studied, in which overall cellular systems consist of different types of cells with different capabilities and configurations \cite{Ghosh:12,Andrews:13,Kountouris:13,Hosseini:13,Hoydis:13}.
One crucial potential for heterogeneous cellular networks is to build so called ``hotspot'' in  the most congested areas within each cell depicted in Fig. \ref{figs:hotspots}, which is beneficial for load valencing, capacity boosting, coverage, and so on \cite{Andrews:12,Hosseini:13,Hoydis:13}.
Although there exist various reasons for considering heterogeneous cellular networks, let us focus on the DoF of heterogeneous cellular networks having hotspots in this subsection.
As shown in Fig. \ref{figs:hotspots}, consider a canonical hotspot model in which the users outside hotspots are served from a macro BS and, on the other hand, the users in each hotspot are served from the micro BS in the same hotspot.
Assume that there are $L$ hotspots in the cell.
Denote the number of antennas at each micro BS by $M_1$ and the number of antennas at the macro BS by  $M_2$.
Also denote the number of users inside each hotspot and the number of users outside hotspots by $N_1$ and $N_2$ respectively.
Each user is assumed to have a single antenna.
Let us focus on the regime that $M_1\leq M_2$ and $N_1\leq N_2$, which is reasonable in practice.

Now again consider the \emph{cell coordination problem}, i.e., how to operate or coordinate this special type of heterogeneous cellular networks in order to maximize its sum DoF.
Recall the results in Theorem \ref{thm:achievable_DoF} and Section \ref{subsec:dof_gain}, suggesting that uplink--downlink operation can improve the DoF of heterogeneous cellular networks. 
We will demonstrate that the same argument holds for the above hotspot network.

First of all, if both micro and macro cells operate as either uplink or downlink, then the sum DoF of the considered hotspot network is upper bounded by 
\begin{align} \label{eq:hotspot_upper}
\min\{L M_1+M_2,LN_1+N_2,\max(LM_1,N_2),\max(LN_1,M_2)\},
\end{align}
where we again use the two-user MIMO IC bound in \cite{Jafar:07}, which corresponds to the model allowing full cooperation between the users in all micro cells and between the micro BSs and also allowing full cooperation between the users in the macro cell.
It might be possible to obtain a tighter bound by considering different types of cooperation, but the above bound is enough to establish an example network demonstrating the DoF gain from uplink--down operation  in the following.

\begin{figure}[t!]
\begin{center}
\includegraphics[scale=0.7]{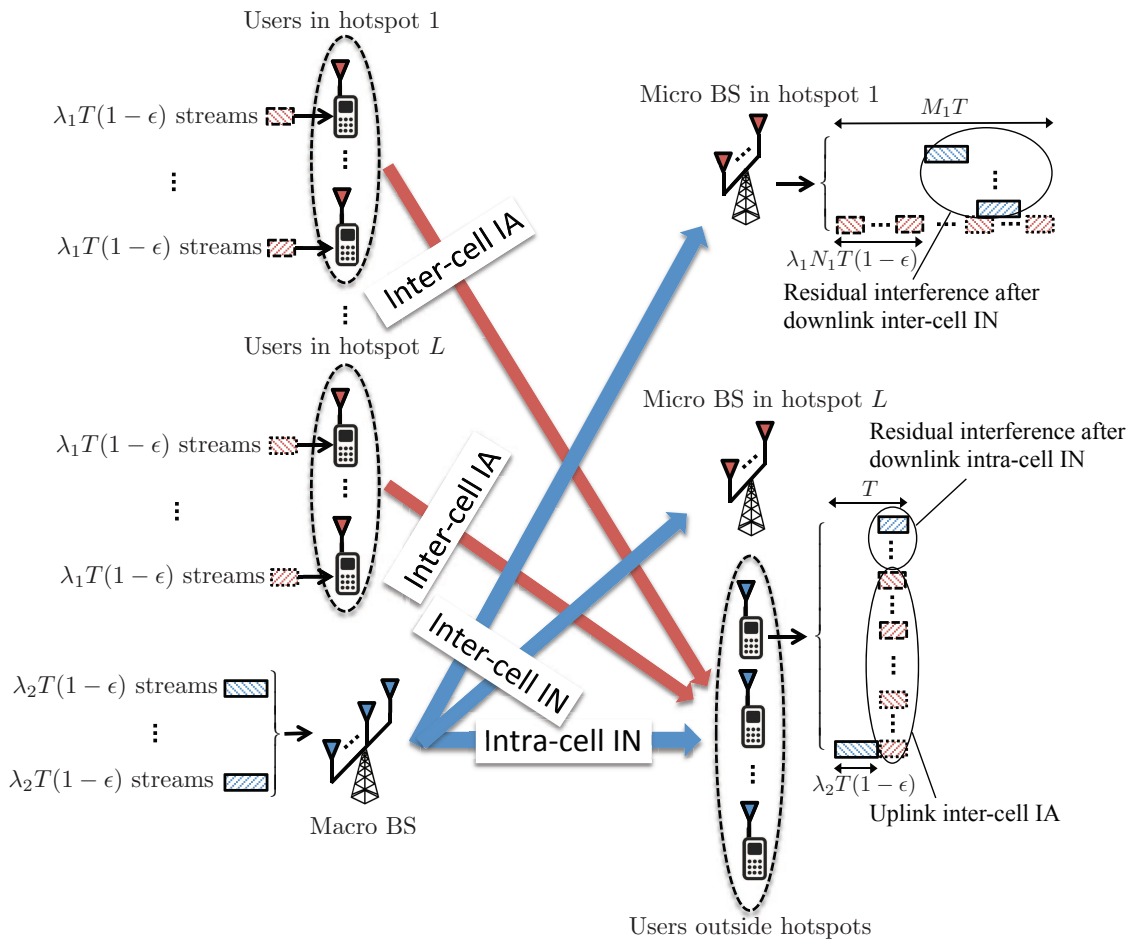}
\caption{IA--IN scheme when all micro cells operate as uplink and the macro cell operates as downlink.}
\label{figs:scheme_hotspot}
\end{center}
\end{figure}

Now operate all micro cells as uplink and the macro cell as downlink depicted in Fig \ref{figs:scheme_hotspot}.
The first IA--IN scheme in Section \ref{subsec:scheme1} can be modified for this case.
Specifically, each user in hotspots transmits $\lambda_1 T(1-\epsilon)$ streams over $T$ time-extended beamforming vectors and the marco BS transmits $\lambda_2 T(1-\epsilon)$ streams to each of the users outside hotspots over $T$ time-extended beamforming vectors, where $\lambda_1, \lambda_2\in(0,1]$ and $\epsilon>0$ is an arbitrarily small constant.
Then, uplink beamforming vectors are set to align inter-cell interference to the users in the macro cell and downlink beamforming vectors are set to null out both inter-cell interference to the micro BSs and intra-cell interference to its serving users.
As seen in Fig. \ref{figs:scheme_hotspot}, each micro BS is able to decode its intended streams almost surely achieving one DoF for each stream by zero-forcing. 
Similarly, each user in the macro cell is able to decode its intended streams almost surely achieving one DoF for each stream by zero-forcing. 
Therefore, as $T$ increases, the following sum DoF is achievable:
\begin{align} \label{eq:hotspot_lower}
\max_{\substack{
            \lambda_1+\lambda_2\leq 1\\
            LN_1\lambda_1\leq M_1\\
            LN_1\lambda_1+N_2\lambda_2\leq M_2}}\{LN_1\lambda_1+N_2\lambda_2\}.
\end{align}

Notice that the above optimization is the same form as in \eqref{eq:maximization1} except that $LN_1$ appears in the object function and the constraints instead of $N_1$. 
Hence we can find the solution of \eqref{eq:hotspot_lower} from $d_{\Sigma,1}$ in Table \ref{table:solution} by substituting $N_1$ with $LN_1$. 

\begin{remark}[Cooperation Between Micro BSs]
If we assume full cooperation between $L$ micro BSs, \eqref{eq:hotspot_lower} is immediately obtained from \eqref{eq:maximization1}. The IA--IN scheme in Fig. \ref{figs:scheme_hotspot} shows that the same sum DoF in \eqref{eq:hotspot_lower} is achievable without joint process sharing their received signals between $L$ micro BSs.  \hfill$\lozenge$
\end{remark}

We can easily find an example that \eqref{eq:hotspot_lower} is strictly greater than \eqref{eq:hotspot_upper}.
For instance, consider the case where $L=2$, $M_1=2$, $M_2=6$, $N_1=3$, and $N_2=4$.
Then, the sum DoFs in \eqref{eq:hotspot_upper} and \eqref{eq:hotspot_lower} are given by $4$ and $\frac{14}{3}$ respectively. 
That is, if we operate this example hotspot network as the conventional downlink, the sum DoF is limited by $4$, which is achievable by only activating the marco cell (The same argument holds for the conventional uplink).
Whereas, if we change the micro cells as uplink, then the sum DoF is improved to $\frac{14}{3}$.
This example suggests that introducing hotspots can improve the sum DoF of cellular networks, but we have to be careful on how to operate or coordinate these heterogeneous cells.

\begin{figure}[t!]
\begin{center}
\includegraphics[scale=0.7]{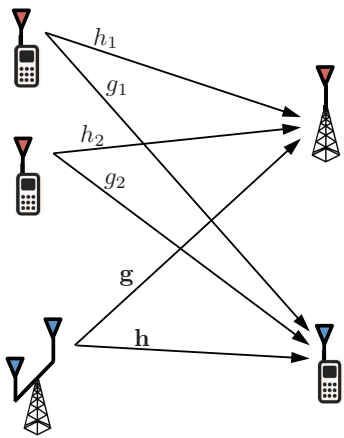}
\end{center}
\vspace{-0.15in}
\caption{Uplink--downlink multiantenna two-cell cellular networks when $M_1=N_2=1$ and $M_2=N_1=2$.}
\label{fig:delayed_csit}
\vspace{-0.1in}
\end{figure}

\subsection{Uplink--Downlink IA With Delayed CSIT} \label{subsec:delayed_csit}
One of the main barriers for implementing IA is for acquiring instantaneous CSI at each transmitter, which is in practice hard to acquire due to the channel feedback delay.
To overcome such limitation of IA using instantaneous CSI, IA using delayed or outdated CSI has been recently studied in the literature \cite{Maddah-Ali:12,Vaze:12,Gou2:12}.
It was originally shown in \cite{Maddah-Ali:12} that completely outdated CSI is still useful for improving DoF of the multiantenna broadcast channel. 
Specifically, delayed CSI was used to align interference in order to exploit received interfering signals as side information. 
The same approach can be applied for uplink--downlink multiantenna two-cell cellular networks.

Consider an example network depicted in Fig. \ref{fig:delayed_csit}, which corresponds to the case where $M_1=N_2=1$ and $M_2=N_1=2$ in Fig \ref{figs:general_i_bc_mac}.
For notational simplicity, we redefine channel coefficients as in Fig. \ref{fig:delayed_csit}.
Let us assume that the users in cell $\alpha$ and BS $\beta$ only knows delayed CSI, i.e., CSI up to time $t-1$ for the transmission at time $t$.
We will show that the sum DoF $\frac{5}{4}$ is achievable using delayed CSI.
Communication takes place over a block of $4$ time slots.
During the transmission block, user $(\alpha, 1)$ transmits two streams $a_1$ and $a_2$, user $(\alpha,2)$ transmits one stream $b_1$, and BS $\beta$ transmits two steams $c_1$ and $c_2$ as follows:
\begin{itemize}
\item At the first time, user $(\alpha,1)$ transmits $a_1$ and user $(\alpha,2)$ transmits $b_1$.
\item At the second time, user $(\alpha,1)$ transmits $a_2$ and user $(\alpha,2)$ transmits $b_1$. 
\item At the third time, BS $\beta$ transmits $[c_1,c_2]^{\dagger}$.
\end{itemize}
Then the received signals of BS $\alpha$ at time $1$, $2$, and $3$ are given by
\begin{align} \label{eq:signal_13_alpha}
h_1[1]a_1+h_2[1]b_1:=L_1(a_1,b_1),\nonumber\\
h_1[2]a_2+h_2[2]b_1:=L_2(a_2,b_1),\nonumber\\
\mathbf{g}[3][c_1,c_2]^{\dagger}:=L_3(c_1,c_2),
\end{align}
respectively, where we omit additive noises in the input--output relation.
Similarly, the received signals of the user in cell $\beta$ at time $1$, $2$, and $3$ are given by
\begin{align}  \label{eq:signal_13_beta}
g_1[1]a_1+g_2[1]b_1:=L_4(a_1,b_1),\nonumber\\
g_1[2]a_2+g_2[2]b_1:=L_5(a_2,b_1),\nonumber\\
\mathbf{h}[3][c_1,c_2]^{\dagger}:=L_6(c_1,c_2).
\end{align}
Then, BS $\alpha$ can decode $a_1$, $a_2$, and $b_1$ if it obtains a linear combination of $(a_1,a_2,a_3)$, linearly independent of $L_1(a_1,b_1)$ and $L_2(a_2,b_1)$, and the user in cell $\beta$ can decode $c_1$ and $c_2$ if it obtains a linear combination of $(c_1,c_2)$, linearly independent of $L_6(c_1,c_2)$. This is possible by transmitting at the fourth time as follows:
\begin{itemize}
\item At the fourth time, user $(\alpha,1)$ transmits $L_7(a_1,a_2)$ and BS $\beta$ transmits $[L_3(c_1,c_2),L_3(c_1,c_2)]^{\dagger}$, where $L_7(a_1,a_2)$ is given by $g_1[1]a_1-\frac{g_2[1]g_1[2]}{g_2[2]}a_2$.
\end{itemize}
Note that user $(\alpha,1)$ can construct $L_7(a_1,a_2)$ and BS $\beta$ can construct $L_3(c_1,c_2)$ using delayed CSI.

The received signal of BS $\alpha$ at time $4$ is given by
\begin{align} \label{eq:signal_4_alpha}
h_1[4]L_7(a_1,a_2)+\mathbf{g}[4][1,1]^{\dagger}L_3(c_1,c_2).
\end{align}
Therefore, by subtracting the effect of $L_3(c_1,c_2)$ from \eqref{eq:signal_4_alpha}, which was received at time $3$, BS $\alpha$ is able to obtain $L_7(a_1,a_2)$ and, as a result, decode $a_1$, $a_2$, $a_3$ from $L_1(a_1,b_1)$, $L_2(a_2,b_1)$, and $L_7(a_1,a_2)$.
The received signal of the user in cell $\beta$ at time $4$ is given by
\begin{align} \label{eq:signal_4_beta}
\mathbf{h}[4][1,1]^{\dagger}L_3(c_1,c_2)+g_1[4]L_7(a_1,a_2).
\end{align}
Hence the user in cell $\beta$ first constructs $L_7(a_1,a_2)=L_4(a_1,b_1)-\frac{g_2[1]}{g_2[2]}L_5(a_2,b_1)$ from $L_4(a_1,b_1)$ and $L_5(a_2,b_1)$, each of which was received at time $1$ and $2$.
Then it subtracts the effect of  $L_7(a_1,a_2)$ from \eqref{eq:signal_4_beta} and, as a result, decode $c_1$ and $c_2$ from $L_3(c_1,c_2)$ and $L_6(c_1,c_2)$.
In conclusion, the sum DoF $\frac{5}{4}$ is achievable and this example demonstrates that delayed CSIT is still useful for  uplink--downlink multiantenna two-cell cellular networks.
Furthermore, if we operate the above example network as the conventional uplink or downlink, then the sum DoF is limited by one even with instantaneous CSIT from the result in \cite{Jafar:07}.
Therefore, it also shows that uplink--downlink operation can improve the sum DoF than the conventional uplink or downlink under the delayed CSIT model. 

\section{Concluding Remarks} \label{sec:conclusion}
In this paper, the sum DoF of uplink--downlink multiantenna two-cell cellular networks has been characterized.
The result demonstrates that, for a broad class of network configurations, uplink--downlink operation can strictly enlarge the sum DoF of multiantenna two-cell cellular networks compared to the conventional uplink or downlink operation.
This DoF improvement basically comes from heterogeneous network environment, especially when the number of antennas at each BS is different from each other.
Recently, for various reasons such as capacity, coverage, load valancing, and so on, heterogeneous cellular networks called ``HetNet'' have been actively studied both in academia and industry.
Therefore we should be more careful for operating such heterogeneous cellular networks consisting of macro BSs with a larger number of antennas and micro BSs with a smaller number of antennas.

\section*{Appendix\\ Optimal $(\lambda_1,\lambda_2)$ and $\max(d_{\Sigma,1}, d_{\Sigma,2})$} \label{APP:linear_programing}
In this appendix, we prove Lemma  \ref{lemma:optimal_sol}.
Recall that
\begin{align} \label{ex:d_sigma1}
d_{\Sigma, 1}=\max_{\substack{
            \lambda_1+\lambda_2\leq 1\\
            N_1\lambda_1\leq M_1\\
            N_1\lambda_1+N_2\lambda_2\leq M_2}}\{N_1\lambda_1+N_2\lambda_2\}
\end{align}
and
\begin{align} \label{ex:d_sigma2}
d_{\Sigma, 2}=\max_{\substack{
            \lambda_1+\lambda_2\leq 1\\
            N_1\lambda_1+N_2\lambda_2\leq M_1\\
            N_2\lambda_2\leq M_2}}\{N_1\lambda_1+N_2\lambda_2\}.
\end{align}

\begin{table}[t!] 
\caption{For given $M_1$, $M_2$, $N_1$, and $N_2$,  $d_{\Sigma, 1}$, $d_{\Sigma, 2}$, and $\max(d_{\Sigma, 1},d_{\Sigma, 2})$.}
\begin{center}
\begin{tabular}{c|c|c|c}  \label{table:solution}
 Case& $d_{\Sigma, 1}$& $d_{\Sigma, 2}$ & $\max(d_{\Sigma, 1},d_{\Sigma, 2})$\\
\hline\hline
1: $M_1\leq M_2\leq N_1\leq N_2$ & $M_2$ & $M_1$ & $M_2$ \\
\hline
2: $M_1\leq M_2\leq N_2\leq N_1$ & $M_2$ & $M_1$ & $M_2$ \\
\hline
3: $M_1\leq N_1\leq M_2\leq N_2$ & $M_2$ & $M_1$ & $M_2$ \\
\hline
4: $M_1\leq N_1\leq N_2\leq M_2$ & $N_2$ & $M_1$ & $N_2$ \\
\hline
5: $M_1\leq N_2\leq M_2\leq N_1$ & $\min(M_2,\frac{N_1 N_2+M_1(N_1-N_2)}{N_1})$ & $M_1$ & $\min(M_2,\frac{N_1 N_2+M_1(N_1-N_2)}{N_1})$ \\
\hline
6: $M_1\leq N_2\leq N_1\leq M_2$ & $\frac{N_1 N_2+M_1(N_1-N_2)}{N_1}$ & $M_1$ & $\frac{N_1 N_2+M_1(N_1-N_2)}{N_1}$ \\
\hline
7: $M_2\leq M_1\leq N_1\leq N_2$ & $M_2$ & $M_1$ & $M_1$ \\
\hline
8: $M_2\leq M_1\leq N_2\leq N_1$ & $M_2$ & $M_1$ & $M_1$ \\
\hline
9: $M_2\leq N_1\leq M_1\leq N_2$ & $M_2$ & $\min(M_1,\frac{N_1N_2+M_2(N_2-N_1)}{N_2})$ & $\min(M_1,\frac{N_1N_2+M_2(N_2-N_1)}{N_2})$\\
\hline
10: $M_2\leq N_1\leq N_2\leq M_1$ & $M_2$ & $\frac{N_1N_2+M_2(N_2-N_1)}{N_2}$ & $\frac{N_1N_2+M_2(N_2-N_1)}{N_2}$ \\
\hline
11: $M_2\leq N_2\leq M_1\leq N_1$ & $M_2$ & $M_1$ & $M_1$ \\
\hline
12: $M_2\leq N_2\leq N_1\leq M_1$ & $M_2$ & $N_1$ & $N_1$ \\
\hline
13: $N_1\leq M_1\leq M_2\leq N_2$ & $M_2$ & $M_1$ & $M_2$ \\
\hline
14: $N_1\leq M_1\leq N_2\leq M_2$ & $N_2$ & $N_2$ & $N_2$ \\
\hline
15: $N_1\leq M_2\leq M_1\leq N_2$ & $M_2$ & $\min(M_1,\frac{N_1 N_2+M_2 (N_2-N_1)}{N_2})$ & $\min(M_1,\frac{N_1 N_2+M_2 (N_2-N_1)}{N_2})$\\
\hline
16: $N_1\leq M_2\leq N_2\leq M_1$ & $M_2$ & $\frac{N_1 N_2+M_2 (N_2-N_1)}{N_2}$ & $\frac{N_1 N_2+M_2 (N_2-N_1)}{N_2}$ \\
\hline
17: $N_1\leq N_2\leq M_1\leq M_2$ & $N_2$ & $N_2$ & $N_2$ \\
\hline
18: $N_1\leq N_2\leq M_2\leq M_1$ & $N_2$ & $N_2$ & $N_2$ \\
\hline
19: $N_2\leq M_1\leq M_2\leq N_1$ & $\min(M_2,\frac{N_1 N_2+M_1(N_1-N_2)}{N_1})$ & $M_1$ & $\min(M_2,\frac{N_1 N_2+M_1(N_1-N_2)}{N_1})$ \\
\hline
20: $N_2\leq M_1\leq N_1\leq M_2$ & $\frac{N_1 N_2+M_1(N_1-N_2)}{N_1}$ & $M_1$ & $\frac{N_1 N_2+M_1(N_1-N_2)}{N_1}$ \\
\hline
21: $N_2\leq M_2\leq M_1\leq N_1$ & $M_2$ & $M_1$ & $M_1$ \\
\hline
22: $N_2\leq M_2\leq N_1\leq M_1$ & $M_2$ & $N_1$ & $N_1$ \\
\hline
23: $N_2\leq N_1\leq M_1\leq M_2$ & $N_1$ & $N_1$ &  $N_1$ \\
\hline
24: $N_2\leq N_1\leq M_2\leq M_1$ & $N_1$ & $N_1$ & $N_1$ \\
\hline
\end{tabular}
\end{center}
\end{table}%

Depending on the relationship between $M_1$, $M_2$, $N_1$, and $N_2$, the solutions of the above two linear programs are represented as in different forms.
Hence we first divide the entire four-parameter space $(M_1,M_2,N_1,N_2)$ into $24$ regimes as shown in Table \ref{table:solution}.\footnote{For simplicity, we allow some overlap between regimes.}
\begin{itemize}
\item Identify a feasible region of $(\lambda_1,\lambda_2)$ for \eqref{ex:d_sigma1}, i.e., the region of $(\lambda_1,\lambda_2)$ satisfying three constraints in \eqref{ex:d_sigma1}.
\item Find $(\lambda_1,\lambda_2)$ maximizing the objective function $N_1\lambda_1+N_2\lambda_2$ among the corner points in the feasible region, which provides $d_{\Sigma,1}$.\footnote{One of the corner points is the solution of a linear program.}
\item Repeat the above two steps for \eqref{ex:d_sigma2}, which provides $d_{\Sigma,2}$.
\item Find $\max(d_{\Sigma,1},d_{\Sigma,2})$.
\end{itemize}

\begin{figure}[t!]
\begin{center}
\includegraphics[scale=1]{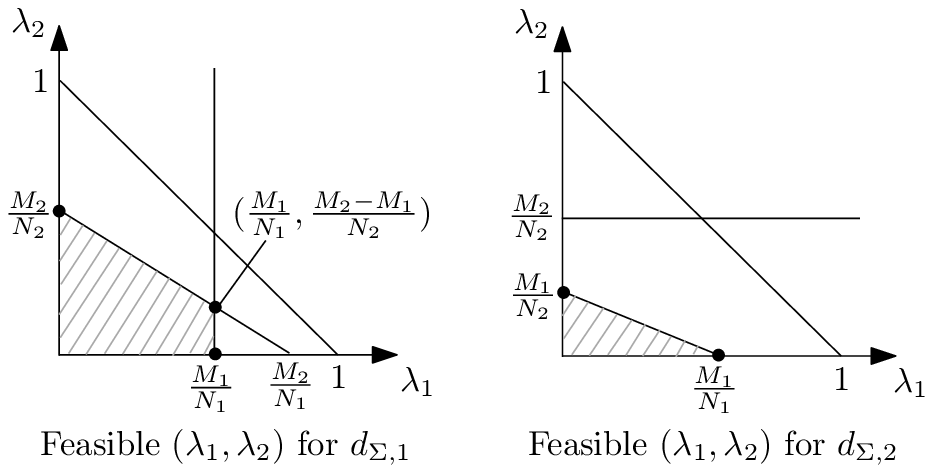}
\end{center}
\vspace{-0.15in}
\caption{Feasible regions of $(\lambda_1,\lambda_2)$ and the corresponding corner points when $M_1\leq M_2\leq N_1\leq N_2$.}
\label{fig:feasible}
\vspace{-0.1in}
\end{figure}

For instance, consider the first regime where $M_1\leq M_2\leq N_1\leq N_2$ in Table \ref{table:solution}. 
Figure \ref{fig:feasible} plots the feasible $(\lambda_1, \lambda_2)$ regions in \eqref{ex:d_sigma1} and \eqref{ex:d_sigma2} for this regime.
For \eqref{ex:d_sigma1}, the first constraint $\lambda_1+\lambda_2\leq 1$ becomes inactive and thus at least one of the three corner points yields the maximum of $N_1\lambda_1+N_2\lambda_2$, which gives $d_{\Sigma,1}=M_2$ when $(\lambda_1,\lambda_2)=(\frac{M_1}{N_1},\frac{M_2-M_1}{N_2})$.
For \eqref{ex:d_sigma2}, on the other hand, only the second constraint $N_1\lambda_1+N_2\lambda_2\leq M_1$ becomes active and at least one of the two corner points yields the maximum, which gives $d_{\Sigma,2}=M_1$ when $(\lambda_1,\lambda_2)=(\frac{M_1}{N_1},0)$ or $(\lambda_1,\lambda_2)=(0,\frac{M_1}{N_2})$.
Hence $\max(d_{\Sigma,1},d_{\Sigma,2})=M_2$ when $M_1\leq M_2\leq N_1\leq N_2$.
In the same manner, we can derive $d_{\Sigma, 1}$ and $d_{\Sigma, 2}$, and $\max(d_{\Sigma, 1},d_{\Sigma, 2})$ for the rest of the regimes in Table \ref{table:solution}.

From Table \ref{table:solution}, $\max(d_{\Sigma,1},d_{\Sigma,2})=d_{\Sigma,1}$ if $M_1\leq M_2$ and $\max(d_{\Sigma,1},d_{\Sigma,2})=d_{\Sigma,2}$ if $M_2\leq M_1$.
Furthermore, $\max(d_{\Sigma,1},d_{\Sigma,2})$ in Table \ref{table:solution} coincides with \eqref{eqn:achievable_DoF} in Theorem \ref{thm:achievable_DoF} for all the regimes. For the regime where $M_1\leq M_2\leq N_1\leq N_2$, for instance, \eqref{eqn:achievable_DoF} is given by
\begin{align}
d_{\Sigma}&=\min\left\{\frac{N_1N_2+M_2(N_2-N_1)}{N_2},M_1+N_2,M_2+N_1,M_2,N_2\right\}\nonumber\\
&=\min\left\{\frac{N_1N_2+M_2(N_2-N_1)}{N_2},M_2\right\}\nonumber\\
&=M_2,
\end{align}
where the second equality follows since $\frac{N_1N_2+M_2(N_2-N_1)}{N_2}=M_2+\frac{N_1(N_2-M_2)}{N_2}\geq M_2$.
In a similar manner, we can prove that $\max(d_{\Sigma,1},d_{\Sigma,2})=d_{\Sigma}$ for the rest of the regimes.
In conclusion, 
\begin{align} 
&\max(d_{\Sigma,1},d_{\Sigma,2})=d_{\Sigma,1}= d_{\Sigma} \mbox{ if }M_1\leq M_2,\nonumber\\
&\max(d_{\Sigma,1},d_{\Sigma,2})=d_{\Sigma,2}= d_{\Sigma} \mbox{ if }M_2\leq M_1,
\end{align}
which completes the proof.


\end{document}